\documentclass[11pt]{article}
\usepackage{etex}
\usepackage[all]{xy}           
\usepackage{amssymb}           
\usepackage{amsmath,amsfonts,amssymb,amsthm,textcomp}
\usepackage{eucal}
\usepackage{epsfig}
\usepackage{tikz-cd}
\usepackage[colorlinks=true]{hyperref}
\usepackage{epic}
\usepackage{color}
\usepackage{fancybox}
\usepackage{wrapfig}
\usepackage[font=footnotesize,labelsep=period]{caption}
\usepackage{float}
\usepackage{booktabs}
\usepackage{fancyhdr}
\usepackage{appendix}
\usepackage{color}
\usepackage{makeidx}

\def\v #1{\vert #1\vert}             

\def\m #1 #2{(-1)^{{\v #1} {\v #2}}} 

\theoremstyle{plain}
\newtheorem{theorem}{Theorem}
\newtheorem{corollary}[theorem]{Corollary}
\newtheorem{proposition}[theorem]{Proposition}
\newtheorem{lemma}[theorem]{Lemma}

\theoremstyle{definition}
\newtheorem{definition}[theorem]{Definition}

\font\frak=eufm10 scaled\magstep1

\def\<#1>{\langle#1\rangle}

\begin{document}

\centerline{\Large \bf Reduction of a Hamilton--Jacobi equation}
\vskip 0.2cm
\centerline{\Large \bf for nonholonomic systems}
\vskip 0.5cm
\vskip 0.5cm

\centerline{O\u{g}ul Esen$^{\dagger}$, Manuel de Le\'on$^{\ddagger}$, V\'ictor Manuel Jim\'{e}nez Morales$^{*}$, Cristina
Sard\'on$^{*}$}
\vskip 0.5cm

\centerline{Department of Mathematics$^{\dagger}$}
\centerline{Gebze Technical University}
\centerline{41400 Gebze, Kocaeli, Turkey.}
\vskip 0.5cm

\centerline{Consejo Superior de Investigaciones Cient\'ificas$^{\ddagger}$}
\centerline{C/ Nicol\'as Cabrera, 13--15, 28049, Madrid. SPAIN}
\vskip 0.5cm

\centerline{Instituto de Ciencias Matem\'aticas, Campus Cantoblanco$^{*}$}
\centerline{Consejo Superior de Investigaciones Cient\'ificas}
\centerline{C/ Nicol\'as Cabrera, 13--15, 28049, Madrid. SPAIN}

\tableofcontents

\begin{abstract}

Nonholonomic mechanical systems have been attracting more interest in recent years because of their rich geometric properties and their applications in Engineering. In all generality, we discuss the reduction of a Hamilton-Jacobi theory for systems subject to nonholonomic constraints and that are invariant under the action of a group of symmetries. We consider nonholonomic systems subject to linear or nonlinear constraints, with different positioning with respect to the symmetries. We describe the reduction procedure first, to later reconstruct solutions in the unreduced picture, by starting from a reduced Hamilton-Jacobi equation. Examples can be depicted in a wide range of scenarios: from free particles with linear constraints, to vehicle motion.

\medskip

\end{abstract}

\section{Introduction}

In this paper we aim at providing a general reduction of a Hamilton--Jacobi formulation for nonholonomic mechanical systems that are invariant under a group of symmetries. The reduction of a Hamilton--Jacobi theory for nonholonomic systems is considered of remarkable interest given the number of wheeled dynamical systems subject to this kind of constraints, and how they play a role in robotics and motion \cite{AA,XLB}.

Nonholonomic constraints are differential constraints depending on the derivatives of the coordinates in the configuration space of the mechanical system. Their study started when the brilliant formalism of Euler and Lagrange failed to be applicable to simple rigid bodies that rolled without sliding. Since then, scientists have made a growing effort in the development of analytic mechanics of nonholonomic systems, which is intimately related with differential geometry.

For example, consider the problem of the rolling disk on a plane $\pi$.
The position of the disc is given by the coordinates $(x,y)$ of its 
point of contact $M$. We call $\varphi$ the angle between the tangent to the disc at the point $P$ and the $Ox$ axis and the angle of inclination $\theta$ between the plane of the disc and $\pi$ \cite{Neimark}.

\begin{figure}[h]
  \centering
   \includegraphics[width=0.5\textwidth]{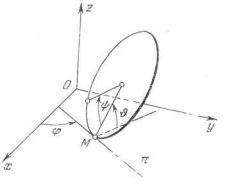} 
   \caption{Rolling disk}
\end{figure}

\noindent
  The condition of rolling without sliding means that the instantaneous velocity of the point of contact of the disc is zero at all times.
The kinematic contraints in this scenario are:
\begin{equation}\label{kincons}
\dot{x}=R\dot{\psi}\cos{\varphi},\quad \dot{y}=R\dot{\psi}\sin{\varphi}
\end{equation}
However, although these conditions must be satisfied, the five coordinates $(x,y,\varphi,\theta,\psi)$ may take all sets of values and, thus, the kinematic constraints (\ref{kincons}) do not impose restrictions on the possible values of these coordinates. This is a characteristic of nonholonomic constraints. These constraints are also nonintegrable.

A particular and common type of nonholonomic constraint is the  linear type in generalized velocities:
\begin{equation}
A_i(q^1,\dots,q^n,t)\dot{q}^i+B(q^1,\dots,q^n,t)=0,
\end{equation}
\noindent
(the time dependence is not strictly necessary).
It is remarkable that these constraints preserve energy, that they are ideal constraints that do not produce work, and that the dynamics dictated by them is constrained to a submanifold \cite{Gantmacher}. This implies that only certain curves correspond to motions of the system; a point of the space representing the position at a given instant of time cannot move in an arbitrary direction.

More general nonholonomic constraint are nonlinear constraints. Consider for example the system depicted in Fig.2

\begin{figure}[h]
  \centering
   \includegraphics[width=0.6\textwidth]{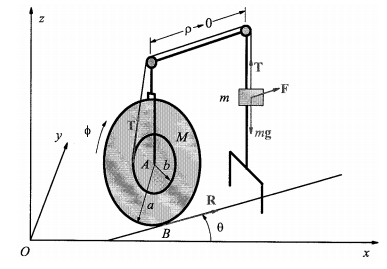} 
   \caption{Appel Hamel system}
\end{figure}

This machine is known as the Appel-Hamel system. It consists of a block of mass that is confined vertically and attached to the end of a thread that passes over two massless pulleys and it is wound round a drum of radius $b$. This is simultaneously attached to a wheel of radius $a$ that rolls horizontally on the $xy$ plane. We can consider an extra force in the picture, applied paralell to the tangent direction of the wheel trace. This force can be realized, for instance, by a rocket oriented on the tangent direction. Here the angle $\phi$ describes the rotation of the wheel on its own plane. The distance between the centers of the two pulleys is $\rho\rightarrow 0$.
So, the nonholonomic constraints are
\begin{equation}
\dot{x}\sin{\theta}-\dot{y}\cos{\theta}=0,\quad \dot{x}\cos{\theta}+\dot{y}\sin{\theta}=a\dot{\phi},\quad \dot{z}=-b\dot{\phi}
\end{equation}
from where we deduce a constrained equation that is nonlinear in the velocities
\begin{equation}
\dot{x}^2+\dot{y}^2=\frac{a^2}{b^2}\dot{z}^2
\end{equation}

It is important to remark that there is a difference between constrained systems and nonholonomic systems. Nonholonomic systems are not Hamiltonian, in the sense that the phase space is just the constraint submanifold and not the cotangent bundle of the configuration manifold. Moreover, the dynamics is explained through an almost Poisson bracket.

Some different frameworks have been proposed to solve the dynamics of linear/nonlinearly nonholonomic systems \cite{BatesSnia,BKMM,EKMR,Koiller}. 
Nonetheless, past attempts to obtain this Hamilton--Jacobi theory for nonholonomic systems were nor effective or very restrictive because they try to adapt the usual proof of Hamilton--Jacobi equations without contraints (using Hamilton's principle).
Usually, the results are valid when the solutions of the nonholonomic problem are also the solutions of the corresponding constrained variational problem \cite{Kozlov,Rumy}.
M. de Le\'on et al. proposed an alternative approach based on geometrical properties. In \cite{LeonMdD}, it is proved that the nonholonomic mechanics can be obtained by projecting the unconstrained dynamics, from the classical framework of the Hamilton-Jacobi formalism, very naturally \cite{IgLeDi07,LeDiVa17}. This approach completes and engages all the partial results that had been resolved until that moment. In parallel, other authors Cari{\~n}ena et al. were developing similar ideas \cite{CaGrMaMaMuRo06}. Let us summarize them.

Recall that the Hamilton--Jacobi equation (HJE) is a partial differential equation for a generating function $S(q^i,t)$ on $Q$ and the time $t$ given by
\begin{equation}\label{tdepHJ}
 \frac{\partial S}{\partial t}+H\left(q^i,\frac{\partial S}{\partial q^i}\right)=0.
\end{equation}
Note that, the generalized momenta do not appear in Eq. (\ref{tdepHJ}), except as derivatives of $S$.
This equation is a necessary condition describing the external geometry in problems of calculus of variations. Hamilton's principal function
$S=S(q^i,t)$, which is the solution of the HJE, and the classical function $H$ are both closely related to the classical action
$S=\int{L dt}$, and it is a generating function
for a family of symplectic flows that describes the dynamics of the Hamilton equations.
 If the generating function is separable in time, then we can make an ansatz $S(q^i,t)=W(q^i)-Et,$
where $E$ is the total energy of the system. Then, HJE in Eq. (\ref{tdepHJ}) reduces to
 \begin{equation}\label{hje}
H\left(q^i,\frac{\partial W}{\partial q^i}\right)=E.
 \end{equation}
Physically, this constant $E$ is identified with the energy of the mechanical system.

Now, let us consider a Hamiltonian system determined by the triple $(T^{*}Q,\omega_Q=-d\theta_{Q},H)$. Here, $\theta_{Q}$ is the canonical Liouville form on $T^{*}Q$, and $\omega_{Q}$ is the canonical symplectic form on $T^{*}Q$. Consider a real valued function $W$ and define a vector field 
\begin{equation}\label{HJtheorem}
 \Gamma_H^{dW}= T_{\pi}\circ \Gamma_H\circ dW.
\end{equation}
on the base manifold $Q$. Here, $\Gamma_H$ is the Hamiltonian vector field generated by the Hamiltonian function $H$. This definition implies the commutativity of the following diagram.
\begin{equation}\label{Xg}
  \xymatrix{ T^{*}Q
\ar[dd]^{\pi_{Q}} \ar[rrr]^{\Gamma_H}&   & &TT^{*}Q\ar[dd]^{T\pi_{Q}}\\
  &  & &\\
 Q\ar@/^2pc/[uu]^{dW}\ar[rrr]^{\Gamma_H^{\sigma}}&  & & TQ }
\end{equation}
According to the geometric Hamilton--Jacobi theory \cite{CaGrMaMaMuRo06}, the vector fields $\Gamma_H$ and $\Gamma_H^{dW}$ are $dW$-related if and only if
 \begin{equation}\label{fulfilled}
  d(H\circ dW)=0.
 \end{equation}
Notice that, equation (\ref{fulfilled}) is merely (\ref{hje}) locally. This means that if the Hamiltonian vector field $\Gamma_{H}$ can be projected into the configuration manifold $Q$ by means of the one-form $dW$, and then the integral curves of the projected
vector field $\Gamma_{H}^{dW}$ can be transformed into integral curves of $\Gamma_{H}$ provided that $W$ is a solution of Eq. (\ref{fulfilled}). This implies that the pullback of the Hamiltonian function $H$ by the one-form section $dW$ is a constant $E$ which is assumed to be the energy of the system. 

It is important to remark that thanks to M. de Le\'on et. al in \cite{deLeMadeDi10}, equation (\ref{fulfilled}) is retrieved in the nonholonomic case as an equality instead of as an inclusion on an ideal space for one-forms, by using the theory of Lie algebroids.

It is possible to generalize this discussion by replacing the exact one-form $dW$ by a closed one-form $\sigma$ defined on $Q$ (and, then, $\sigma \left( Q \right) \subseteq T^{*}Q$ is a Lagrangian submanifold of $\left( T^{*}Q , \omega_{Q} \right)$). 
This modern geometric interpretation has brought about the theory of Lagrangian submanifolds as an important source of new results and insights \cite{LeDiVa17}. Recall that a Lagrangian submanifold $L$ of a symplectic manifold $\left( M , \omega \right)$ is a submanifold of $M$ such that for all $x \in L$
$$ T_{x} L = \left( T_{x}L\right)^{\bot},$$
where $\left( T_{x}L\right)^{\bot}$ is the symplectic orthogonal of $T_{x}L$ respect to the symplectic form $\omega$.\\

Furthermore, the geometric setting for the Hamilton--Jacobi theory will be particulary useful in the identification of conserved quantities in mechanical systems, which may be possible even when the mechanical problem itself cannot be solved completely. These conserved quantities are related with the unknown initial conditions describing a family of solutions of the Hamilton--Jacobi equation, which are parametrized in a particular convenient way.
These conserved quantities are useful for reducing our system to another system in a lower dimensional manifold. For this, we need to construct a momentum map $J$ (satisying some invariance conditions) from $T^*Q$ to the dual of the Lie algebra $\mathfrak{g}$ of the symmetry group \cite{MaWe74} that exploits the symmetry of Hamiltonian and Lagrangian systems following the modern work of J.E. Marsden and A. Weinstein (1974) \cite{MaWe74}. The pre-image $J^{-1}(\mu)$ of a dual element under the momentum map is a submanifold of momentum phase space and there exists a reduction of the symplectic form such that we find a reduced symplectic manifold \cite{Ho11,MaRa13}.

The reduction of the geometric Hamilton-Jacobi theory for unconstrained systems was achieved through coisotropic reduction \cite{LeDiVa17}. On the other hand, a first more general and effective reduction of a Hamilton--Jacobi theory for nonholonomic constraints was introduced in the aforementioned paper \cite{IgLeDi07}. There, it is presented for the particular case in which the configuration space is a fibration over another manifold  $\rho:Q\rightarrow N$, and the constraints are given by the horizontal subspaces of a connection $\Gamma$ on the fibration $\rho$. In this case, the original nonholonomic system is equivalent to another one whose configuration manifold is the base of the fibration and, in addition, it is subject to an external force \cite{LeonDMDD}. These are known as Caplygin systems \cite{Cap1,Cap2}, and for them, we will define a connection $\Gamma$, whose horizontal distribution is $\mathcal{H}$ and $V_{\rho}=\ker{T\rho}$ is the vertical distribution. The connection allows the Whitney decomposition $TQ=\mathcal{H}\oplus V_{\rho}$.

In this paper, we aim at generalizing the reduction of a Hamilton--Jacobi theory for nonholonomic systems with linear and nonlinear constraints, by a group of symmetries. We will make a difference among different types of these systems depending on the relative positioning between the constraint functions and the symmetries.
We have three cases: the so called pure kinematic nonholonomic systems, the horizontal nonholonomic systems, and a more general case that involves these two.
Notice that our results will be displayed in the Lagrangian and Hamiltonian part equivalently, and that we will retrieve the results in \cite{LeonMdD} as a particular case of our general picture. 

 We will characterize our nonholonomic system by a distribution $F$ in $T^{*}Q$ whose annihilator $F^{o}$ is the set of reaction forces on $T^{*}Q$  \cite{Cap1,Cap2}.
The dynamics of the nonholonomic system is defined by the constraint functions of a submanifold $N\subset TQ$ in the Lagrangian picture, and $M\subset T^{*}Q$ in the Hamiltonian case. We also consider the case when the constraints (linear or not) are bracket generating or completely nonholonomic, and so, we extend to the nonlinear case.

The structure of the paper is as follows: In section 2 we recall the fundamentals of nonholonomic mechanics in the Lagrangian and Hamiltonian part, and present the Hamilton--Jacobi equation for systems with external forces, with linear and nonlinear constraints in the Lagrangian and the Hamiltonian picture. Here we have included an extension of the result by M. de Le\'on et al \cite{IgLeDi07,deLeMadeDi10} and T. Ohsawa et al \cite{ToOsBlo} to the case of nonlinear constraints. Section 3 is devoted to the theory of reduction of nonholonomic systems according to the positioning of the symmetries with respect to the constraint functions, in the three aforementioned cases. We display the Lagrangian and Hamiltonian formalisms for each of the three cases. In section 4 we introduce the actual reduction of a Hamilton--Jacobi equation for nonholonomic systems with symmetries. The pure kinematic and horizontal cases are exhaustively discussed and a general case including these two is also depicted, and all the three cases are studied in the Lagrangian and Hamiltonian side as well. To conclude the paper, we add section 5 with examples that shed some light on the long and nontrivial proof and results given in this manuscript. One of them is a real representant of the true applications of nonholonomic mechanics, as it is the study of the dynamics of vehicles like a two wheeled carriage.

For simplicity, we hereafter assume all mathematical objects to be real, smooth, and globally defined.
This will permit us to omit several minor technical problems so as to highlight the main aspects of our
results.

\section{Fundamentals of nonholonomic mechanics}

Consider an $n$-dimensional configuration manifold $Q$ and choose coordinates as $(q^i)$. We consider its tangent and cotangent bundles $TQ$ and $T^{*}Q$ with canonical projections $\tau_Q$, and $\pi_Q$, respectively. 
In bundle coordinates, the tangent space is locally coordinated by $(q^i,\dot{q}^j)$
and the phase space $T^{*}Q$ by the Darboux' coordinates $(q^i,p_j)$. The phase space $T^{*}Q$ is equipped with a canonical one-form $\theta_Q$, called the Liouville--Euler one-form. Minus of the exterior derivative of $\theta_{Q}$ is the canonical symplectic two-form $\omega_{Q}$ on $T^*Q$, In coordinates, $\theta_Q$ manifests its semi-basic character whereas the symplectic two-form is constant skew symmetric two-form 
\begin{equation}
 \theta_{Q}=p_idq^i, \qquad \omega_{Q}=dq^i\wedge dp_i
\end{equation}

We are assuming that all the Lagrangians are regular. This means the Legendre transformation 
\begin{equation} \label{LegTrf}
\mathbb{F}L:TQ\longrightarrow T^*Q:\left(q^i,\dot{q}^j\right)\mapsto \left( q^i,\frac{\partial L}{\partial \dot{q}^j} \right) 
\end{equation}
is a local diffeomorphism. Then, the pullback of the symplectic two-form $\omega_Q$ on $T^*Q$ back to $TQ$ by means of the mapping $\mathbb{F}L$ is a symplectic two-form on $TQ$. This symplectic two-form is called the Poincar\'e-Cartan two-form, and we denote it by $\omega_L$. As a result, in this section, the tangent bundle $(TQ,\omega_L)$ is a symplectic manifold. The energy of the Lagrangian $L$ is defined to be 
\begin{equation} \label{energy}
E_L(q,\dot{q})=\dot{q}^i\frac{\partial L}{\partial \dot{q}^i}-L(q,\dot{q}).
\end{equation}
and the symplectic equation
\begin{equation} \label{HamEq-EL}
\iota_{\Gamma_L}\omega_L=dE_L
\end{equation}
uniquely determines a vector field $\Gamma_L$, known as the Euler-Lagrange vector field, that generates the motion. For a detailed description see \cite{LeRo11}.

It is possible to study constrained systems in both the Lagrangian and the Hamiltonian picture. In the following section, we start with constrained Lagrangian systems.
  
\subsection{Constrained Lagrangian systems} \label{constrLagsec}

Consider now a nonholonomic Lagrangian system consisting of a regular Lagrangian function $L$ on $TQ$ and a submanifold $N \subseteq TQ$ of codimension $k$ in $TQ$ defining the nonholonomic constraints of the system. Then, $N$ may be locally described in terms of independent constraint functions $\{\psi^a \}_{a=1, \dots , k}$ in the following way
\begin{equation} \label{Npre}
N=\left\{(q^i,\dot{q}^j)\in TQ: \psi^a(q,\dot{q})=0\right\}.
\end{equation}
We will always assume that $\tau_{Q} \left( N \right) = Q$ or, equivalently, the constraints are purely kinematic.\\

According to the d'Alembert-Chetaev principle, the equation of motion of the constrained Lagrangian system is given by
\begin{equation}\label{nhlag8}
\frac{d}{dt}\left(\frac{\partial L}{\partial \dot{q}^i}\right)-\frac{\partial L}{\partial q^i}=\lambda_a \frac{\partial \psi^a}{\partial \dot{q}^i}, \quad \psi^a(q,\dot{q})=0,
\end{equation}
where $\lambda_a$ are the Lagrange multipliers to be determined (see for instance \cite{Vakonomic,leonmarrdmdd}).\\
The submanifold $N$ defines a distribution $F_L$ on $TQ$ along $N$ whose annihilator $F_L^{o}$ plays the role of ``reaction forces" on $TQ$. Notice that any distribution on a manifold $\mathcal{M}$ along a submanifold $\mathcal{S}\subseteq \mathcal{M}$ can be naturally depicted as a generalized vector subbundle of the tangent bundle $T\mathcal{M}$ over $\mathcal{S}$. Along this paper we will use both terminologies.\\
These forces are one-forms on $TQ$ along the submanifold $N$, and can be understood by taking their values in the subbundle
\begin{equation} \label{codist}
 F_L^{o}=S^{*}(TN^{o})=\langle \frac{\partial \psi^a}{\partial \dot{q}^i} dq^i \rangle,
\end{equation}
of $T^{*}_{N}TQ$, where $S$ is the vertical endomorphism, i.e., the $(1,1)$-tensor field on the base manifold $TQ$ with expression
\begin{equation} \label{vertendo}
S=\frac{\partial}{\partial \dot{q}^i}\otimes dq^i,
\end{equation}
Notice that, the codistribution $F_L^o$ is generated by semi-basic one-forms. Thus, the equation of motion associated to the constrained Lagrangian system can be written as \cite{leonmarrdmdd}:
\begin{equation}\label{nonhollag}
\iota_{\Gamma_{L,N}}\omega_L-dE_L\in S^{*}(TN^{o}),\quad \Gamma_{L,N}|_{N}\in TN.
\end{equation}
It is remarkable that any solution $\Gamma_{L,N}$ of Eq. (\ref{nonhollag}) is a SODE, since the $1-$forms in $F_{L}^{o}$ are semibasic.\\
Note that if the submanifold $N$ is particularly chosen to be $TQ$, then the constrained system (\ref{nonhollag}) reduces to the unconstrained one in Eq. (\ref{HamEq-EL}).\\
Nonetheless, the existence and uniqueness of a vector field $\Gamma_{L,N}$ satisfying Eq. (\ref{nonhollag}) cannot be guaranteed. The system  has a solution if and only if the value of the unconstrained vector field $\Gamma_{L}$ falls into the sum $TN+F_L^{\bot}$ at every point 
in $N$. Here, $F_L^{\bot}$ denotes the symplectic complement of the distribution $F_L$ with respect to the symplectic two-form $\omega_L$. To guarantee existence and uniqueness of $\Gamma_{L,N}$ we will impose two conditions:
\begin{itemize}
\item[(i)] Admissibility condition: for each $x \in N$
\begin{equation} \label{admissibility condition}
\dim \left( T_{x} N \right)^{o} = \left(\dim F_{L}^{o}\right)_{x}
\end{equation}
\item[(ii)] Compatibility condition: for each $x \in N$,
\begin{equation}\label{compatibility condition}
TN \cap F_L^{\bot } =\{0\}
\end{equation}
\end{itemize}
Admissibility condition simply implie that the restriction of the vertical endomorphism $S^{*}: TN^{o} \rightarrow F_{L}^{o}$ is an isomorphism or, equivalently, the family of $1-$forms $\{ \frac{\partial \psi^a}{\partial \dot{q}^i} dq^i  \}$ is linearly independent.\\
To understand the compatibility condition, let us define the following matrix:
$$ \mathcal{C}= \left(\mathcal{C}^{ab} \right) = \left( Z^{a} \left( \psi^{b}\right) \right),$$
where each $Z^{b}$ is the (local) vector field on $TQ$ such that $S^{*} \left( d\psi^{b} \right)$ is the image of $Z^{b}$ by the musical isomorphism $\omega_{L}^{\flat}$, i.e.,
$$\omega_{L} \left( Z^{b} , X \right) = \left[S^{*}\left( d \psi^{b} \right)\right]\left( X \right) , \  \forall X \in \mathfrak{X} \left( TQ \right).$$
So, the compatibility condition is equivalent to the regularity of $\mathcal{C}$  (independently of the choice of the constraint functions). Let us consider the Hessian matrix
$$ \left( W_{ij}\right) = \left( \dfrac{\partial^{2} L}{\partial \dot{q}^{i} \dot{q}^{j}} \right).$$
Then, 
$$\mathcal{C}^{ab} = - W^{ij}\dfrac{\partial \psi^{a}}{\partial \dot{q}^{i}}\dfrac{\partial \psi^{b}}{\partial \dot{q}^{j}},$$
where $\left( W^{ij} \right)$ is the inverse matrix of $\left(W_{ij}\right)$. By using this fact we could deduce that the compatibility condition is satisfied if the Lagrangian $L$ is of mechanical type, i.e., $L = T-V$, where $T$ is the kinetic energy function derived from a metric on $Q$, and $V$ is the potential energy defined on $Q$ (see for instance \cite{CaLeMaDi98}). On the other hand, if $N$ is a vector subbundle of $TQ$ admissibility condition is easily satisfied (see \ref{linearconstraints23323}). From now on, we will assume that these two conditions hold.\\

In this case, the tangent space at each point in the constrained submanifold $N$ can be written as a direct sum of $TN$ and $F_L^{\bot}$, that is 
\begin{equation}\label{projections1}
\left. TTQ\right\vert _{N}=TN\oplus F_L^{\bot}. 
 \end{equation}
This geometry enables us to define two complementary projectors namely 
$$\mathcal{P}:TTQ \vert_N \mapsto TN, \qquad \mathcal{Q}:TTQ\vert_N \mapsto F_L^{\bot}.$$ 
It is straightforward to check that $\Gamma_{L,N}=\mathcal{P}(\Gamma_L)$ is a solution of the constrained system. This says that the projection $\mathcal{P}$ is mapping the Euler-Lagrange vector field $\Gamma_L$ to uniquely a solution $\Gamma_{L,N}$ of Eq. (\ref{nonhollag}), that is 
\begin{equation} \label{soln-nHHe}
\mathcal{P}(\Gamma_L)=\Gamma_{L,N}.
\end{equation}

A constraint function is said to be ideal if the work of the forces of reaction of the constraint is equal to zero. A constrained Lagrangian system is called ideal if all the constraints are ideal. Notice that a constrained Lagrangian system is ideal if and only if the Liouville vector field $\Delta = \dot{q}^i{\dfrac{\partial}{\partial\dot{q}^i}}$ is tangent to the constraint submanifold $N$. In local charts, this reads that, for every ideal constraint function $\psi^a$, the term $\dot{q}^i{\dfrac{\partial}{\partial\dot{q}^i}}$ vanishes identically on the constraint submanifold $N$. An important fact related with the notion of ``ideal" is the conservation of energy. In fact, if the constrained Lagrangian system is ideal the energy is a conserved quantity for the solution $\Gamma_{L,N}$ of Eq. (\ref{nonhollag}), i.e.,
\begin{equation}\label{eq13dasd}
 \Gamma_{L,N} \left( E_{L} \right) = 0.
\end{equation}
In fact, for each $a$
$$\left[ S^{*}\left( d \psi^{a}\right) \right]\left( \Gamma_{L,N}\right) =  d\psi^{a} \left( S \left( \Gamma_{L,N} \right) \right) = d\psi^{a} \left( \Delta \right) = 0.$$
So, $\Gamma_{L,N} \in F_{L}$. Then, by contracting Eq. (\ref{nonhollag}) we obtain Eq. (\ref{eq13dasd}). Subsequently, we will assume this condition holds.\\
Notice that, if the constraint functions depend on the velocity variables linearly, then $N$ is ideal. We discuss this particular case in detail in forthcoming subsections.

Consider the distribution $F_L$ on $TQ$ along $N$ whose annihilator is the codistribution $F_L^o$ defined in Eq. (\ref{codist}). Notice that, at each point $v$ in $N \subseteq TQ$, the subspace $(F_L)_v$ is a coisotropic subspace of the symplectic vector space $(T_vTQ,{\omega_L} \left( v \right))$. If the compatibility condition (\ref{compatibility condition}) is satisfied, then the distribution $\mathcal{H}$, which is defined to be the intersection 
\begin{equation} \label{H}
\mathcal{H}:=TN \cap F_L,
\end{equation}
becomes a symplectic subbundle of the symplectic vector bundle $T_{N}TQ$. Converse of this assertion is also true. That is if $\mathcal{H}$ is a symplectic subbundle then the constraint system satisfies the compatibility condition \cite{leonmarrdmdd}. Observe that $\omega_{L}$ can be naturally restricted to $\mathcal{H}$ as a map $\omega_{\mathcal{H}}$ from $N$ to $\mathcal{H}^{*}$, i.e., for all $v_{q} \in N$,
\begin{equation}\label{naturally}
\omega_{\mathcal{H}}\left( v_{q} \right) : \mathcal{H}_{v_{q}} \rightarrow \mathbb{R},
\end{equation}
is the restriction of $\omega_{L}\left( v_{q} \right)$ to $ \mathcal{H}_{v_{q}}$. Here $\omega_{\mathcal{H}}$ will be called $2-$form on $\mathcal{H}$. Analogously, we restrict $dE_{L}$ to $\mathcal{H}$ as the $1-$form $dE_L\vert_\mathcal{H}$ on $\mathcal{H}$. Here, $E_L$ is the energy exhibited in Eq. (\ref{energy}). In this case, the unique solution $\Gamma_{H}$ of the symplectic equation 
$$\iota_{\Gamma_{H}} \omega_{\mathcal{H}}=dE_L\vert_{\mathcal{H}}$$
is the vector field $\Gamma_{L,N}$ generating the constrained dynamics. To prove this fact we only have to restrict Eq. (\ref{nonhollag}) to $\mathcal{H}$ and use the nondegenerancy property of $\omega_{\mathcal{H}}$ (for the linear case see \cite{BatesSnia,CuKeSnBa95}). Further, we have the following direct sum decomposition 
\begin{equation}
T(TQ)\vert _N = \mathcal{H}\oplus \mathcal{H}^{\bot}.
\end{equation}
This decomposition enables us to define two new complementary projectors 
$$\mathcal{\bar{P}}:TTQ \vert_N \mapsto \mathcal{H}, \qquad \mathcal{\bar{Q}}:TTQ\vert_N \mapsto \mathcal{H}^{\bot}.$$ 
See that, we have the following relation between unconstrained Euler-Lagrange vector field $\Gamma_L$ and the constrained one $\Gamma_{L,N}$ given by 
\begin{equation}
\mathcal{\bar{P}}(\Gamma_L)=\Gamma_{L,N}.
\end{equation}
\bigskip
\noindent \textbf{Constrained Lagrangian systems with linear constraints.}\label{linearconstraints23323}\\
A particularly important example of nonholonomic Lagrangian system is when the constraint manifold $N$ is a linear subbundle of the tangent bundle $TQ$. In this case $N$ will be denoted by $N_{\ell}$. We may prove that $N_{\ell}$ is (locally) genereated by linear constraint functions in the velocities, i.e., 
\begin{equation} \label{constLin}
\psi^a(q,\dot{q})=\psi_i^a(q)\dot{q}^i.
\end{equation}
So that we have 
\begin{equation} \label{N2}
N_{\ell}=\left\{(q^i,\dot{q}^j)\in TQ: \psi_i^a(q)\dot{q}^i=0\right\}.
\end{equation}
Taking into account that the functions $\psi^a$ are \textit{fibre-wise} linear we can define the $1-$forms
\begin{equation} \label{barpsi}
\overline{\psi}^a(q)=\psi_i^a(q) dq^i \in \Lambda^1(Q)
\end{equation}
It is easy to realize that the codistribution $N_{\ell}^{o}$ on $Q$ given by the annihilator of $N_{\ell}$ is (locally) generated by the $1-$forms $\overline{\psi}^{a}$.\\
Then, by pulling the one-forms $\overline{\psi}^a$ back to $TQ$ with the help of the tangent bundle projection $\tau_Q$, we arrive at one-forms
\begin{equation} \label{barpsilift}
\tau_Q^*\overline{\psi}^a(q)=\psi_i^a(q) dq^i = S^{*}\left( d \psi^a \right)\in \Lambda^1(TQ)
\end{equation} 
Hence,
$$\tau_{Q}^{*}\left( N_{\ell}^{o}\right) = S^{*}\left( TN_{\ell}^{o}\right).$$
This immediately implies that, in this case, the admissibility condition is satisfied.\\
Note that, even though the one form sections presented in Eq. (\ref{barpsi}) and Eq. (\ref{barpsilift}) locally look identical, they are not so. The former ones $\overline{\psi}^a$ are one-form sections on $Q$ taking values in $T^*Q$ whereas the latter ones $\tau_Q^*\overline{\psi}^a$ are defined on $TQ$ and they are taking values in $T^*TQ$.

Eventually, we are ready to write the equation of motion for the case of linear constraint functions by substituting the local generators of the codistribution $S^{*}(TN_{\ell}^{o})$. Locally, the system (\ref{nonhollag}) reduces to

\begin{equation}\label{nhlag2}
\frac{d}{dt}\left(\frac{\partial L}{\partial \dot{q}^i}\right)-\frac{\partial L}{\partial q^i}=\lambda_a \psi^a_{i}(q), \quad \psi^a_i(q)\dot{q}^i=0,
\end{equation}
where $\lambda_a$ are the Lagrange multipliers to be determined.

\subsection{Constrained Hamiltonian dynamics} \label{constHamSec}

Consider a Hamiltonian system $(T^*Q,\omega_Q,H)$ defined on the cotangent bundle equipped with the canonical symplectic form. Denote the Hamiltonian vector field of $H$ by $\Gamma_{H}$, i.e., $\Gamma_{H}$ is the unique solution of the Hamilton equation
\begin{equation}\label{HamiltstandardEQ13234}
\iota_{\Gamma_{H}}\omega_{Q} = dH.
\end{equation}
Let $M$ be a submanifold of $T^*Q$ of codimension $k$ defining the constraints of the system. Then, $M$ may be described in terms of independent constraint functions $\{\Psi^a \}_{a=1, \dots , k}$ in the following way
\begin{equation}
M=\left\{z\in T^*Q: \Psi^a(z)=0\right\}.
\end{equation}
The equation of motion associated to the constraint Hamiltonian system is defined to be 
\begin{equation} \label{nHELag}
\left. \left(\iota_{\Gamma_{H,M}}\omega_{Q} -dH\right) \right\vert _{M}\in F^{o}, \qquad
\left. \Gamma_{H,M} \right\vert _{M} \in TM,
\end{equation}
where $F^{o}$ is just an arbitrary codistribution on $T^{*}Q$. We will define $F$ as the distribution on $T^{*}Q$ such that $F^{o}$ is the annihilator of $F$. Furthermore, $F^{\bot}$ is the symplectic orthogonal of $F$ with respect of the symplectic form $\omega_{Q}$. Here, $F^{o}$ plays the role of reaction forces. \\
Again, to guarantee the existence and uniqueness of a vector field $\Gamma_{H,M}$ satisfying 
Eq. (\ref{nHELag}) we need to impose admissibility and compatibility conditions
\begin{itemize}
\item[(i)] Admissibility condition: $dim \left( T_{x} M \right) = dim \left( F_{x} \right), \ \forall x \in M$

\item[(ii)] Compatibility condition: $TM \cap F^{\perp} = \{0 \}.$
\end{itemize}
Then the tangent space at each point $z$ in the constrained submanifold $M$ can be written as a direct sum of $T_{z}M$ and $F^{\bot}_{z}$, that is 
\begin{equation}\label{projections}
\left. TT^*Q\right\vert _{M}=TM\oplus F^{\bot}. 
 \end{equation}
and we can define two complementary projectors as well
$$\widehat{\mathcal{P}}:TT^*Q \vert_M \mapsto TM, \qquad \widehat{\mathcal{Q}}:TT^*Q\vert_M \mapsto F^{\bot}.$$ So, $\Gamma_{H,M}=\widehat{\mathcal{P}}(\Gamma_H)$ is a solution of the constrained system (\ref{nHELag}). This says that the projection $\mathcal{P}$ is mapping the Hamiltonian vector field $\Gamma_H$ to a solution $\Gamma_{H,M}$ of Eq. (\ref{nHELag}), that is 
\begin{equation} \label{soln-nHHeHamilt}
\Gamma_{H,M}=\widehat{\mathcal{P}}(\Gamma_H).
\end{equation}
\noindent
If the codistribution $F^{o}$ is generated by the set of one-forms $(\sigma^{a})$ where $a$ runs from $1$ to the dimension of $F^{o}$, say $k$, using the inverse of the musical isomorphism $\omega^{\sharp}_{Q}$, we compute the symplectic gradients of these generators and obtain a basis $\{Z^a\}$ for the symplectic orthogonal $F^\bot$ of the distribution $F$. If we denote $\Gamma_{H}$ as the unconstrained Hamiltonian vector field, a solution of the constrained system takes the form 
\begin{equation} \label{-soln-nHHe}
\Gamma_{H,M}=\Gamma_{H}+\lambda_{a}Z^a, 
\end{equation}
where $\{\lambda_{a}\}$ are Lagrange multipliers.

\subsection{Legendre transformations of constrained Lagrangian systems.}
As a particular case, let us now apply the Legendre transformation to the constraint Lagrangian system presented in Eq. (\ref{nonhollag}) with linear constraint functions Eq. (\ref{constLin}). Note that, the constraint submanifold $N_{\ell}\subset TQ$ in Eq. (\ref{N2}) is mapped to the constraint submanifold 
\begin{equation} \label{LT}
M=\left \{(q,p)\in T^*Q: \Psi^a(q,p)=\psi^a_{i}(q)\frac{\partial H}{\partial p_i}(q,p) =0 \right \}
\end{equation}
of the cotangent bundle $T^*Q$. In order to complete the Legendre transformation of the constraint Lagrangian system presented in Eq. (\ref{nonhollag}) to the cotangent bundle, we now transform codistribution on $TQ$ presented in Eq. (\ref{codist}) to the cotangent bundle to arrive at a codistribution $F^{o}$ defined on the cotangent bundle $T^*Q$. To do this, we compute the Legendre transformation of $S^{*}(TN_{\ell}^{o})$ as follows 
\begin{eqnarray} 
F^{o}&=&(\mathbb{F}L^{-1})^*(S^{*}(TN_{\ell}^{o}))=\left\langle (\mathbb{F}L^{-1})^*\tau_Q^*\overline{\psi}^a \right\rangle \notag \\ 
&=&\left\langle (\tau_Q \circ \mathbb{F}L^{-1})^*\overline{\psi}^a\right\rangle \notag \\ 
&=&\left\langle \pi_Q^*\overline{\psi}^a\right\rangle.
\end{eqnarray}
Note that the relation $\tau_Q \circ \mathbb{F}L^{-1}=\pi_Q$ is the manifestation of the following commutation relation of the fiber derivative.
\begin{equation}
\xymatrix{ TQ
 \ar[dr]_{\tau_Q} \ar@<1ex>[rr]^{\mathbb{F}L} &&T^{*}Q \ar[ll]^{\mathbb{F}L^{-1}}\ar[dl]^{\pi_Q}\\
  & Q&}
 \end{equation}
In coordinates, $F^{o}$ is computed to be 
\begin{equation} \label{F^{o}Lag}
F^{o}=\left\langle\sigma^a:=\pi_Q^*\overline{\psi}^a= \psi^a_{i}(q) dq^i \right\rangle.
\end{equation}

With these definitions, $M$ and $F^{o}$ satisfy the admissibility condition. In Darboux' coordinates, the constraint Hamiltonian system (\ref{nHELag}) is computed to be
\begin{equation} \label{nhHamEqcoor}
\dot{q}^i=\frac{\partial H}{\partial p_i}, \qquad \dot{p}_i=-\frac{\partial H}{\partial q^i}+\beta_a\psi^a_i(q), \qquad  \psi^a_i(q)\frac{\partial H}{\partial p_i}=0.
\end{equation}
Here, the Hamiltonian function is in the form of Eq. (\ref{LT}). 
It is possible to determine the Lagrange multipliers $\beta_a$ in the equations of the motion by simply taking derivative of the constrained with respect to time. By taking the isomorphic image of the space $F^{o}$ in Eq. (\ref{F^{o}Lag}), we compute the distribution 
\begin{equation}
F^\bot=\left\langle Z^a=\psi^a_{i}(q)\dfrac{\partial}{\partial p_i} \right \rangle
\end{equation}
over the cotangent bundle $T^*Q$.

 \subsection{Hamilton--Jacobi equation for systems with external forces} \label{sec-HJforExt}
Let us now introduce the geometric formulation of the Hamilton-Jacobi problem with external forces. A detailed development of this theory can be found in \cite{IgLeDi07}.\\
First, we will present the geometric version of the Halmilton-Jacobi without external forces (or with the external forces equal to zero). Let $\sigma$ be a $1-$form on $Q$ and the Hamiltonian vector field $\Gamma_{H}$ defined in Eq. (\ref{HamiltstandardEQ13234}). We will define a vector field $\Gamma_{H}^{\sigma} \in \frak \mathfrak{X} \left(Q \right)$ as follows
\begin{equation}
\Gamma_{H}^\sigma:=T\pi_{Q} \circ \Gamma_{H}\circ \sigma.
\end{equation}
So, the following diagram 
\begin{equation}\label{Xg21StandardH}
  \xymatrix{ T^{*}Q
\ar[dd]^{\pi_{Q}} \ar[rrr]^{\Gamma_{H}}&   & &T\left( T^{*}Q \right)\ar[dd]^{T\pi_{Q}}\\
  &  & &\\
 Q\ar@/^2pc/[uu]^{\sigma}\ar[rrr]^{\Gamma_{H}^\sigma}&  & & TQ }
\end{equation}
is commutative.
\begin{theorem}\label{standarHJEq234Hamilton}
Let $\sigma$ be closed $1-$form on $Q$. Then the following conditions are equivalent:

\begin{itemize}
\item[(i)] $\Gamma_{H}^\sigma$ and $\Gamma_{H}$ are $\sigma$-related

\item[(ii)] $d(H \circ \sigma) =0$.
\end{itemize}
\end{theorem}
\begin{definition}
A $1-$form $\sigma$ on $Q$ satisfying conditions of the Theorem \ref{standarHJEq234Hamilton} will be called solution for the standard Hamilton-Jacobi problem given by the Hamiltonian $H$.
\end{definition}

For the Hamiltonian counterpart consider $X$ a vector field on $Q$ and the Euler-Lagrange vector field $\Gamma_{L}$ defined in Eq. (\ref{HamEq-EL}). Then, we can define a vector field $\Gamma_{L}^{X} \in \frak \mathfrak{X} \left(Q \right)$ as follows
\begin{equation}
\Gamma_{L}^X:=T\tau_{Q} \circ \Gamma_{L}\circ X.
\end{equation}
By taking into account that the energy $E_{L}$ of the Lagrangian $L$ is given by $H \circ \mathbb{F}L$ we have the following result.
\begin{theorem}\label{standarHJEq234Lagrangian}
Let $X$ be vector field on $Q$ such that $\mathbb{F}L \circ X$ is a closed $1-$form on $Q$. Then the following conditions are equivalent:

\begin{itemize}
\item[(i)] $\Gamma_{L}^X$ and $\Gamma_{L}$ are $X$-related

\item[(ii)] $d(E_L \circ X) =0$.
\end{itemize}
\end{theorem}
\begin{definition}
A vector field $X$ on $Q$ satisfying conditions of the Theorem \ref{standarHJEq234Lagrangian} will be called solution for the standard Hamilton-Jacobi problem given by the regular Lagrangian $L$.
\end{definition}

Let us now present the geometric formulation of the Hamilton-Jacobi theory for mechanical systems with external forces. The external forces will be depicted by a semibasic $1-$form $\beta$ (i.e., $\beta$ vanishes over vertical tangent vectors) on $Q$. So, symplectic equation associated to $\beta$ is
\begin{equation}
\iota_{\Gamma_{H,\beta}}\omega_{Q}-dH=\beta
\end{equation} 
Note that $\Gamma_{H,\beta}$ is a vector field on $T^*Q$ and it is equal to $\Gamma_H$ if $\beta$ vanishes identically. Let $\sigma$ be a closed one-form on $Q$, and introduce a vector field on $Q$ as follows
\begin{equation}
\Gamma_{H,\beta}^\sigma:=T\pi_Q\circ \Gamma_{H,\beta} \circ \sigma.
\end{equation}
Here is the Hamilton-Jacobi theorem for the present framework.
\begin{theorem} \label{HJforConst}
The following conditions are equivalent:
\begin{itemize}
\item [(i)]$\Gamma_{H,\beta}^\sigma$ and $\Gamma_{H,\beta}$ are $\sigma$ related.
\item [(ii)]$d\left(H\circ \sigma \right)=-\sigma^*\beta$.
\end{itemize}
\end{theorem}

\begin{definition}
A $1-$form $\sigma$ on $Q$ satisfying conditions of the Theorem \ref{HJforConst} will be called solution for the Hamilton-Jacobi problem given by the Hamiltonian $H$ and the external forces $\beta$.
\end{definition}

Let us now write this theorem on the symplectic structures defined on the tangent bundles. Assume the following symplectic relation on $TQ$ equipped with the symplectic form $\omega_L$
\begin{equation}
\iota_{\Gamma_{L,\alpha}}\omega_L-dE_L=\alpha.
\end{equation} 
Here, $\Gamma_{L,\alpha}$ is the vector field generating the dynamics, and $\alpha$ is a semibasic one-form on $TQ$. Let $X$ be a vector field on $Q$ such that $\mathbb{F}L\circ X$ is closed one-form on $Q$, and define the following vector field 
\begin{equation}
\Gamma_{L,\alpha}^X=T\tau_Q \circ \Gamma_{L,\alpha} \circ X,
\end{equation}
where $T\tau_Q$ is the tangent lift of $\tau_Q$. We have the following
\begin{theorem} \label{HJforConstLag}
The following conditions are equivalent:
\begin{itemize}
\item [(i)]$\Gamma_{L,\alpha}^X$ and $\Gamma_{L,\alpha}$ are $X$-related.
\item [(ii)]$d \left(E{_L}\circ X\right)=-X^*\alpha$.
\end{itemize}
\end{theorem}
\begin{definition}
A vector field $X$ on $Q$ satisfying conditions of the Theorem \ref{HJforConstLag} will be called solution for the Hamilton-Jacobi problem given by the regular Lagrangian $L$ and the external forces $\alpha$.
\end{definition}
\subsection{Hamilton--Jacobi equation for nonholonomic systems with linear constraints}

Let us introduce the geometric Hamilton--Jacobi theory in the framework of Lagrangian system with linear constraints. Accordingly, we start with a  regular Lagrangian $L$ on $TQ$ and the existence of nonholonomic constraints defined by the vector subbundle $N_{\ell}$ given in \eqref{N2}. As discussed previously, see Eq. (\ref{barpsi}), the constraint functions $\psi ^a_i(q)$ lead to the determination of the set of differential one-forms $\overline{\psi^a}$ on $Q$. Note that, the image space of these one-forms determine a subbundle $N_{\ell}^0$ of the cotangent bundle $T^*Q$ annihilating $N_{\ell}$, and  
define an ideal 
\begin{equation}
{\mathcal I}(N_{\ell}^0)=\left\{\beta_a\wedge \overline{\psi ^a}: \beta_a\in \Lambda^s(Q)  \right\}
\end{equation}
 of the exterior algebra $\Lambda(Q)$. In this setting, a vector field $X$ on $Q$ is called a characteristic vector field of the ideal by satisfying $\iota_X(\overline{\psi^a})=0$ for all $\overline{\psi^a}$. A characteristic vector field $X$ of ${\mathcal I}(N_{\ell}^0)$ preserves the ideal, that is, $\iota_X{\mathcal I}(N_{\ell}^0)\subset {\mathcal I}(N_{\ell}^0)$. Notice that a vector field $X$ taking values in the constraint subbundle $N_{\ell}$ is a characteristic vector field of the ideal. 
 
Let us enunciate here the Hamilton--Jacobi theorem in this prescribed setting. Consider $X$ be vector field on $Q$ such that $X(Q) \subset N_{\ell}$, and the constrained Euler-Lagrange vector field $\Gamma_{L,N}$ defined in Eq. (\ref{nonhollag}). We will define a vector field $\Gamma_{L,N_{\ell}}^{X} \in \frak X \left(Q \right)$ as follows
\begin{equation}
\Gamma_{L,N_\ell}^X:=T\tau_{Q} \circ \Gamma_{L,N_\ell}\circ X.
\end{equation}
So, the following diagram 
\begin{equation}\label{Xg21}
  \xymatrix{ N_\ell\subset TQ
\ar[dd]^{\tau_{Q}} \ar[rrr]^{\Gamma_{L,N_\ell}}&   & &TTQ\ar[dd]^{T\tau_{Q}}\\
  &  & &\\
 Q\ar@/^2pc/[uu]^{X}\ar[rrr]^{\Gamma_{L,N_\ell}^X}&  & & TQ }
\end{equation}
is commutative.
 
\begin{theorem}\label{nhhj21}
Let $X$ be vector field on $Q$ such that $X(Q) \subset N_{\ell}$ and
$d(\mathbb{F}L \circ X)\in {\mathcal I} (N_{\ell}^0)$. Then the following
conditions are equivalent:

\begin{itemize}
\item[(i)] $\Gamma_{L,N_\ell}^X$ and $\Gamma_{L,N_\ell}$ are $X$-related

\item[(ii)] $d(E_L \circ X) \in N_\ell^0$.
\end{itemize}
\end{theorem}

Let us now introduce a special case of nonholonomic linear constraints which turns condition $d(E_L \circ X) \in N_\ell^0$ into the identity
$$d \left(E_L \circ X \right) = 0.$$

\begin{definition} \label{complete-nonholonomic}
A distribution $\mathcal{D} \leq TQ$ is said to be completely nonholonomic (or bracket-generating) if
$\mathcal{D}$ along with all of its iterated Lie brackets $\left[\mathcal{D}, \mathcal{D}\right], \left[ \mathcal{D}, \left[ \mathcal{D} , \mathcal{D} \right] \right], \hdots$ spans the tangent bundle $TQ$
\end{definition}
An important result related with this kind of distribution is the following
\begin{theorem}[Chow-Rashevskii's theorem]\label{LordChowTheorem}
Let $Q$ be a connected differentiable manifold. If a distribution $\mathcal{D} \leq TQ$ is completely nonholonomic, then any two points on $Q$ can be joined by a horizontal path $\gamma$, i.e., the derivative of $\gamma$ is tangent to $\mathcal{D}$.

\end{theorem}
A detailed proof of this theorem for regular distributions can be found in \cite{Montg}. For singular distributions see \cite{Harms}. A particularly important consequence of this theorem is given by the following result (\cite{ToOsBlo}).\\
\begin{proposition}\label{Importantconsequuence2134}
Let $Q$ be a connected differentiable manifold and $\mathcal{D} \leq TQ$ be a completely nonholonomic distribution. Then there is no non-zero exact one-form in the annihilator $\mathcal{D}^{o} \leq T^{*}Q$.
\end{proposition}
It is important to remark that this proposition can be analogously proved for the case in which $\mathcal{D}$ is a singular distribution by taking into account the Chow-Rashevskii's theorem for singular distributions (\cite{Harms}).\\
So, as an immediate consequence we have the next theorem
\begin{theorem}\label{nhhj21corollary4}
Assume that the distribution defined by the constraint vector bundle $N_{\ell}$ is completely nonholonomic. Let $X$ be vector field on $Q$ such that $X(Q) \subset N_{\ell}$ and $d(\mathbb{F}L \circ X)\in {\mathcal I} (N_{\ell}^0)$. Then, the following
conditions are equivalent:

\begin{itemize}
\item[(i)] $\Gamma_{L,N_\ell}^X$ and $\Gamma_{L,N_\ell}$ are $X$-related

\item[(ii)] $d(E_L \circ X) =0$.
\end{itemize}
\end{theorem}
This result was proved by T. Ohsawa and A. Bloch in \cite{ToOsBlo}.\\

Let us now present this same theorem in its Hamiltonian counterpart \cite{IgLeDi07}, notice that the constraint submanifold is defined to be $M=\mathbb{F}L(N_\ell)$ as in Eq. (\ref{LT}).\\
As above, for a fixed $1-$form $\sigma$ on $Q$ such that $\sigma \left( Q \right) \subseteq M$ we will define a vector field $\Gamma_{H,M}^{\sigma}$ on $Q$ by satisfying the following diagram,

\begin{equation}\label{Xg22}
  \xymatrix{ M\subset T^* Q
\ar[dd]^{\pi_{Q}} \ar[rrr]^{\Gamma_{H,M}}&   & &TT^* Q\ar[dd]^{T\pi_{Q}}\\
  &  & &\\
 Q\ar@/^2pc/[uu]^{\sigma}\ar[rrr]^{\Gamma_{H,M}^{\sigma}}&  & & TQ }
\end{equation}

\begin{theorem}\label{nhhj1}
Let $\sigma$ be a 1-form on $Q$ such that $\sigma(Q) \subset
{M}$ and $d\sigma\in {\mathcal I} (N_{\ell}^0)$. Then the following
conditions are equivalent:

\begin{itemize}
\item[(i)] $\Gamma_{H,M}^{\sigma}$ and $\Gamma_{H,M}$ are $\sigma$-related

\item[(ii)] $d(H \circ \sigma) \in N_{\ell}^0$.
\end{itemize}
\end{theorem}

Notice that this theorem is only proved for a constraint manifold $M \subseteq T^{*}Q$ which comes from a linear constraint manifold $N_{\ell}\subseteq TQ$, via a regular Lagangian $L$. Finally, let us present a Hamiltonian version of Theorem \ref{nhhj21corollary4}.
\begin{theorem}\label{nhhj1corollary4again}
Assume that the distribution defined by the constraint vector bundle $N_{\ell}$ is completely nonholonomic. Let $\sigma$ be a 1-form on $Q$ such that $\sigma(Q) \subset
{M}$ and $d\sigma\in {\mathcal I} (N_{\ell}^0)$. Then, the following
conditions are equivalent:

\begin{itemize}
\item[(i)] $\Gamma_{H,M}^{\sigma}$ and $\Gamma_{H,M}$ are $\sigma$-related

\item[(ii)] $d(H \circ \sigma) =0$.
\end{itemize}
\end{theorem}

\subsection{Hamilton--Jacobi equation for nonholonomic systems with non-linear constraints}\label{firstresut23}

As a first aim of this paper, we will generalize the Hamilton-Jacobi theorems (\ref{nhhj21}) and (\ref{nhhj1}) for the case of non-linear constraints. Recall the constraint submanifold $N$ in Eq. (\ref{Npre}) with (possibly nonlinear) constraint functions $\psi^a$ and recall the codistribution $F^{o}_{L}$ defined in Eq. (\ref{codist}) generated by the differential one-forms $S^{*} \left( d \psi^{a} \right)=\dfrac{\partial \psi^{a}}{\partial \dot{q}^i} dq^i$ on $TQ$. We introduce the ideal $\mathcal{I} \left( F^{o}_{L} \right)$ generated by these one-forms, i.e., the ideal $\mathcal{I} \left( F^{o}_{L} \right)$ is defined as follows
\begin{equation}\label{definitiononeoftheamountofideals23}
{\mathcal I}(F_{L}^o)=\left\{\beta_a\wedge S^{*} \left( d \psi^{a} \right): \beta_a\in \Lambda^s(TQ)  \right\}
\end{equation}

Consider the Euler-Lagrange vector field $\Gamma_{L,N}$ for the constraint system. Let $X$ be vector field on $Q$ such that $X(Q) \subset N$ and define a vector field $\Gamma_{L,N}^{X}$ by the composition $T{\tau_Q}\circ \Gamma_{L,N}\circ X$. Notice that, this definition coincides with the one presented in the diagram \eqref{Xg21} by replacing the linear constraint submanifold $N_\ell$ with an arbitrary submanifold $N$.\\

\begin{theorem}\label{FirstThofred}
Let $X$ be vector field on $Q$ such that $X(Q) \subset N$ and $\left(X \circ \tau_{Q}\right)^{*} \omega_{L} \in \mathcal{I} \left(F^{o}_{L}\right)$. Then, the following conditions are equivalent:

\begin{itemize}
\item[(i)] $\Gamma_{L,N}^{X}$ and $\Gamma_{L,N} $ are $X$-related.
\item[(ii)] $d \left( E_{L} \circ \left( X \circ \tau_{Q}\right) \right) \in F^{o}_{L}$
\end{itemize} 
\end{theorem}
We first present the following two lemmas that we shall need while proving the Theorem (\ref{FirstThofred}). 
\begin{lemma}\label{omega-T}
Consider the symplectic manifold $(TQ,\omega_L)$ with a regular Lagrangian $L$, then
\begin{equation}  
\omega _{L}(U,W)=0, \qquad \text{if}  \quad U,W\in \ker T\tau _{Q}.
\end{equation}
\end{lemma}

\begin{proof}
Let $\alpha _{L}$ be the Poincare-Cartan one-form on $Q$. Then, by using that $\alpha_{L}$ is a semibasic $1-$form, we have that for each two vertical (local) vector fields $W$ and $U$ on $TQ$
\begin{eqnarray*}
\omega _{L}\left( U,W\right) &=&-d\alpha _{L}\left( U,W\right) \\ &=&-U \left(
\alpha_{L} \left( W\right)\right) +W\left( \alpha_{L} \left( U\right)\right)  + \alpha_{L} \left( [U,W]\right) =0.
\end{eqnarray*}
Notice that the Lie bracket of vector fields preserves vertical vector fields.\\

\end{proof}
\begin{lemma}\label{inva}
The identities
\begin{equation}
\left( X \circ \tau_{Q}\right)^{*}F^{o}_{L} =  F^{o}_{L}, \qquad \left( X \circ \tau_{Q}\right)_{*}F_{L} =  F_{L}
\end{equation}
hold along $X$.
\end{lemma}
\begin{proof}
The mapping $ X \circ \tau_{Q} $ reduces to the identity mapping when it is restricted to the base manifold $Q$. This gives that, along $X$, the pullback of $F^{o}_{L}$ turns out to be equal to $F^{o}_{L}$  since it is composed of semi-simple one-forms on $TQ$. It is possible to observe this in terms of coordinates as follows
\begin{eqnarray*}
\left( X \circ \tau_{Q} \right)^{*} \left[ \left(\dfrac{\partial \psi^{a}}{\partial \dot{q}^i}\right)dq^i \right] & = & \left(\dfrac{\partial \psi^{a}}{\partial \dot{q}^i} \circ X \circ \tau_{Q}\right)d\left( q^i \circ  X \circ \tau_{Q}\right)\\
& = & \left(\dfrac{\partial \psi^{a}}{\partial \dot{q}^i} \circ X \circ \tau_{Q}\right)d q^i.
\end{eqnarray*}
Take a generic vector field $Z$ taking values in $F_L$ and consider its restriction to the image space of $X$. According to the Lemma \eqref{inva}, if $\alpha$ is in $F^{o}_{L}$ then so does $\left( X \circ \tau_{Q}\right)^{*}\alpha$. This manifests that its pushforward $\left( X \circ \tau_{Q}\right)_{*}Z$ is an element of $F_L$. 
\end{proof}

\noindent
\begin{proof} [Proof of the Theorem (\ref{FirstThofred}).]
First notice that the constrained Euler-Lagrange vector field $\Gamma_{L,N}$ takes values in $F_L$. So that it is a characteristic vector field of the ideal $\mathcal{I} \left( F^{o}_{L} \right)$ presented in Eq. (\ref{definitiononeoftheamountofideals23}), i.e.
$$ \iota_{\Gamma_{L,N}} \mathcal{I} \left( F^{o}_{L} \right) \subseteq \mathcal{I} \left( F^{o}_{L} \right).$$
In the statement of the theorem we have assumed that $\left(X \circ \tau_{Q}\right)^{*} \omega_{L} \in \mathcal{I} \left(F^{o}_{L}\right)$ therefore further we have that
\begin{equation}\label{NR135}
\iota_{\Gamma_{L,N}} \left(X \circ \tau_{Q}\right)^{*} \omega_{L} \in F^{o}_{L}.
\end{equation}
This reads that by taking the pullback of the both hand side of the equality in \eqref{nonhollag} with the mapping $X \circ \tau_{Q}$ we arrive at that
\begin{equation}\label{NEWNR135}
 \left(X \circ \tau_{Q} \right)^{*}\left[ \iota_{\Gamma_{L,N}} \omega_{L} \right] \in F^{o}_{L}
 \end{equation}
if, and only if 
\begin{equation} \left( X \circ \tau_{Q}\right)^*dE_{L} = d \left( E_{L} \circ X \circ \tau_{Q} \right) \in F^{o}_{L}.
 \end{equation}
This gives that the second condition $\left( ii \right)$ can be equivalently written as Eq. (\ref{NEWNR135}). 
\\ \\
\noindent $\left( i \right) \implies \left( ii \right): $ Let $W$ be an arbitrary vector field on $TQ$, and compute the following 
\begin{eqnarray*}
\left( X\circ \tau _{Q}\right) ^{\ast }\left( \iota _{\Gamma
_{L,N}}\omega _{L}\right) (W)  &=& \tau _{Q}^{\ast
}X^{\ast }\left( \iota _{\Gamma _{L,N}}\omega _{L}\right) (W)
\\
&=& X^{\ast }\left( \iota _{\Gamma _{L,N}}\omega _{L}\right)
\left(T\tau _{Q}\cdot W\right)
\\
&=& \iota_{\Gamma _{L,N}^{X}}X^{\ast }\omega _{L} \left( T\tau
_{Q}\cdot W\right)
\\
&=&X^{\ast }\omega _{L}\left( \Gamma _{L,N}^{X},T\tau _{Q}\cdot W\right)  \\
&=&X^{\ast }\omega _{L}\left( T\tau _{Q}\circ \Gamma _{L,N}\circ X,T\tau
_{Q}\cdot W\right)  
\\
&=&\tau _{Q}^{\ast }X^{\ast }\omega _{L}\left( \Gamma _{L,N}\circ X,W\right) 
\\
&=&\left( X\circ \tau _{Q}\right) ^{\ast }\omega _{L}\left( \Gamma
_{L,N}\circ X,W\right),
\end{eqnarray*}%
where we have used the $X$-relatedness of $\Gamma_{L,N}^{X}$ and $\Gamma_{L,N} $ in the third line whereas we have used the definition of $\Gamma_{L,N}^{X}$ in the fifth line. Restrict this calculation on the image space of $X$ so that
\begin{eqnarray*}
\left( X\circ \tau _{Q}\right) ^{\ast }\left( \iota _{\Gamma
_{L,N}}\omega _{L}\right) (W)  &=&\left( X\circ \tau _{Q}\right)
^{\ast }\omega _{L}\left( \Gamma _{L,N}\circ X,W\right)  \\
&=& \iota _{\Gamma _{L,N}}\left( X\circ \tau _{Q}\right) ^{\ast
}\omega _{L}(W),
\end{eqnarray*}%
for an arbitrary vector field $W$. So,
\[
\left( X\circ \tau _{Q}\right) ^{\ast }\left( \iota _{\Gamma _{L,N}}\omega
_{L}\right) =\iota _{\Gamma _{L,N}}\left( X\circ \tau _{Q}\right) ^{\ast
}\omega _{L}
\]%
along $X$. Since the term in the right hand side of this identity is an element of $F_{L}^{o}$ we see that 
\[
\left( X\circ \tau _{Q}\right) ^{\ast }\left( \iota _{\Gamma _{L,N}}\omega
_{L}\right) \in F_{L}^{o}.
\]
This is the equivalent representation of the condition $\left( ii \right)$ given in Eq. (\ref{NEWNR135}). So we have proved that $\left(i\right)$ implies $\left( ii \right)$. 
\\ \\
\noindent
$\left( ii \right) \implies \left( i \right) :$ 
Consider the constrained Euler-Lagrange vector field $\Gamma_{L,N}$ and an arbitrary vector field $V$ taking values in $F_L$. Instead of the whole $TQ$, we are only interested in the image space of $X$. Introduce the following vector fields along $X$,
\begin{eqnarray} \label{defnvf}
\Lambda:&=&\Gamma_{L,N} - T \left( X \circ \tau_{Q} \right) \cdot \Gamma_{L,N} , \\
\Sigma:&=&V - T \left( X \circ \tau_{Q} \right) \cdot V .
\end{eqnarray}
Notice that both of the vectors $\Lambda(X(q))$ and $\Sigma(X(q))$ are in the kernel of $T\tau_Q$ so that, according to Lemma (\ref{omega-T}),  at $X(q)$, we compute
\begin{equation} \label{keromega}
\omega_L(\Lambda,\Sigma)\vert_{X(q)}=0.
\end{equation}
Substitute now the definitions of the vector fields in Eq. (\ref{defnvf}) into the identity (\ref{keromega}). This reads that 
\begin{eqnarray} \label{zeros}
\omega_L(\Lambda,\Sigma)\vert_{X(q)}&=& \omega_L(\Gamma_{L,N},V)\vert_{X(q)} \notag
-\omega_L(\Gamma_{L,N},T \left( X \circ \tau_{Q} \right) \cdot V )\vert_{X(q)}\notag
\\ &=&  - \ \omega_L(T \left( X \circ \tau_{Q} \right) \cdot \Gamma_{L,N},V)\vert_{X(q)} \notag
\\ &=&  + \ \omega_L(T \left( X \circ \tau_{Q} \right) \cdot \Gamma_{L,N},T \left( X \circ \tau_{Q} \right) \cdot V)\vert_{X(q)}.
\end{eqnarray}
Recall that, we are assuming that $\left(X \circ \tau_{Q}\right)^{*} \omega_{L} \in \mathcal{I} \left(F^{o}_{L}\right)$ which implies Eq. (\ref{NR135}). This makes the term in the third line of the calculation (\ref{zeros}) zero. The condition $\left( ii \right)$ given in the form Eq. (\ref{NEWNR135}) makes the second term in the right hand side of the first line of the calculation (\ref{zeros}) zero. Eventually, we left with the following identity 
\begin{equation}
\omega_L(\Gamma_{L,N}-T \left( X \circ \tau_{Q} \right) \cdot \Gamma_{L,N},V)\vert_{X(q)}=0
\end{equation}
for an arbitrary $V$ in $F_L$. Hence, we have proved that
$$ \left(\Gamma_{L,N}-T ( X \circ \tau_{Q})\cdot \Gamma_{L,N}\right) \vert_{X(q)} \in  F_{L}^{\bot}.$$
From the compatibility condition (\ref{compatibility condition}), we have that $\Gamma_{L,N}-T ( X \circ \tau_{Q})\cdot \Gamma_{L,N}$ is identically zero along $X$. This is nothing but the condition $(i)$.
\end{proof}

\begin{definition}
A vector field $X$ on $Q$ satisfying conditions of the Theorem \ref{FirstThofred} will be called solution for the constrained Hamilton-Jacobi problem given by the constraint manifold $N \subseteq TQ$ and a regular Lagrangian $L$.
\end{definition}

Notice that, as an immediate consequence of condition $(ii)$ of Theorem \ref{FirstThofred}, we   have that any solution $X$ for the standard Halmilton-Jacobi problem (see Theorem \ref{standarHJEq234Lagrangian}) which satisfies that $X \left( Q \right)  \in N$ is a solution of the constrained Hamilton-Jacobi problem given by the constraint manifold $N \subseteq TQ$ and a regular Lagrangian $L$.

\begin{corollary}\label{funnycorollary123}
Let $X$ be a vector field on $Q$ in the conditions of Theorem \ref{FirstThofred} such that $T X \left( TQ \right) \subseteq F_{L}$. Then, the following conditions are equivalent:

\begin{itemize}
\item[(i)] $\Gamma_{L,N}^{X}$ and $\Gamma_{L,N} $ are $X$-related.
\item[(ii)] $d \left( E_{L} \circ  X  \right) =0.$
\end{itemize} 
\begin{proof}
Notice that, by using that $\tau_{Q}$ is a submersion, the $1-$form $d \left( E_{L} \circ  X  \right)$ is zero if, and only if,
$$d  \left( E_{L} \circ \left( X \circ \tau_{Q}\right) \right) =0.$$
Finally, it is easy to check that the only way in which the $1-$form $d \left( E_{L} \circ \left( X \circ \tau_{Q}\right) \right)$ is in the annihilator $F_{L}^{o}$ is that it is cancelled. This is an immediate consequence of the condition $T X \left( TQ\right) \subseteq F_{L}$. 
\end{proof}
\end{corollary}

Consider a vector field $X$ on $Q$ is under conditions of Corollary \ref{funnycorollary123}. Then, by taking into account condition $\left(X \circ \tau_{Q}\right)^{*} \omega_{L} \in \mathcal{I} \left(F^{o}_{L}\right)$ we have that $TX \left(N \right) \subseteq F_{L}$ implies that
$$ X^{*} \omega_{L} = 0.$$
In fact, for each $a$

\begin{eqnarray*}
S^{*} \left( d \psi^{a} \right) \left( T X \left(  \dfrac{\partial}{\partial q^{i}} \right) \right) & = & S^{*} \left( d \psi^{a} \right) \left(  \dfrac{\partial}{\partial q^{i}} +  \dfrac{\partial   X^{k}}{\partial q^{i}} \dfrac{\partial}{\partial \dot{q}^{k}}\right)\\
& = & d \psi^{a} \left(  \dfrac{\partial}{\partial \dot{q}^{i}}  \right)\\
& = &\dfrac{\partial \psi^{a}}{\partial \dot{q}^{i}}.
\end{eqnarray*}
Hence, $TX \left(TQ \right) \subseteq F_{L}$ is equivalent to
$$ F_{L}^{0} = 0 \ \Leftrightarrow \ F_{L} = T_{N} \left( TQ \right).$$
Thus, this condition implies that the constraint functions do not produce reaction forces. In fact, it is enough to find some vector field $X$ satifying that $T X \left( TQ \right) \subseteq F_{L}$

Indeed, the equation of motion associated to the constraint system is defined to be 
\begin{equation}\label{nLagQuasi342}
 \iota_{\Gamma_{L,N}}\omega_{L} -dE_{L}=0, \qquad
\left. \Gamma_{L,N} \right\vert _{N} \in TN.
\end{equation} 

Let us now give another consequence of Theorem \ref{FirstThofred} which can give us more interesting examples. In order to do that, we will use Chow-Rashevskii's theorem (see \ref{LordChowTheorem}) for the non-linear case.\\
Notice, as a first difference with the linear case, $F_{L}$ is just a distribution on $TQ$ along $N$ (it is not defined on the whole $TQ$). Thus, we will solve this problem by defining the following family $\mathcal{S} $ of local $1-$forms on $TQ$
$$\mathcal{S} := \{S^{*}\left( \sigma \right) \in \Lambda^1_{local}\left( TQ \right) \ : \ \sigma_{\vert TN} \equiv 0 \}.$$

So, we define the codistribution $G_{L}^{o}$ on $TQ$ as the codistribution generated by $\mathcal{S}$, i.e., for each $v_{q} \in T_{q} Q$ the fiber ${G_{L}^{o}}_{\vert v_{q}}$ at $v_{q}$ of $G_{L}^{o}$ is the vector subspace of $T^{*}_{v_{q}}\left( TQ \right)$ generated by all the evaluations of the (local) $1-$forms in $\mathcal{S}$ at $v_{q}$. So, $G_{L}$ will be the distribution on $TQ$ such that its annihilator is $G_{L}^{o}$.\\
Then, $G_{L}$ is a smooth (possibly singular) distribution on $TQ$. Notice that, $G_{L}$ can be alternatively defined as the smooth distribution generated by the (local) vector fields on $TQ$ such that its restrictions to $N$ are tangent to $F_{L}$.\\
Observe that far from $N$ the fibers can be the whole tangent space of $TQ$. Another remark is that there exists, at least, one global $1-$form contained in $\mathcal{S}$: the zero section of $\pi_{Q}$.\\
Finally, for all $v_{q} \in N$
$${G_{L}}_{\vert v_{q}} = {F_{L}}_{\vert v_{q}}.$$
In fact, each covector $\sigma_{v_{q}} \in F_{L}^{o}$ can be written as a combination
$$ \sigma_{v_{q}} = \lambda_{a} \left[ S^{*} \left( d\psi^{a} \right) \right] \left( v_{q} \right),$$
and all the local $1-$forms $ S^{*} \left( d\psi^{a} \right)$ are contained in $\mathcal{S}$.\\ 
Let us consider the family of (local) vector fields $\mathcal{S}_{N}^{\sharp}$ tangents to $G_{L}$ whose domains have non-empty intersection with $N$. Then, $F_{L}$ is the distribution generated by $\mathcal{S}_{N}^{\sharp}$ along $N$.\\
Thus, $F_{L}$ is called completely nonholonomic if the reiterated Lie bracket of vector fields in $\mathcal{S}_{N}^{\sharp}$ generate $T_{N} \left( TQ \right)$. In this case we will say that $\left[F_{L},F_{L}\right], \left[ F_{L}, \left[ F_{L},F_{L} \right] \right], \hdots$ spans the vector bundle $T_{N}\left(TQ \right)$.
Obviously, $G_{L}$ is completely nonholonomic if, and only if, $F_{L}$ is completely nonholonomic.\\
we are now ready to give a generalization of Theorem \ref{nhhj1corollary4again}.
\begin{theorem}\label{AppChowtheore,}
Let $X$ be a vector field on $Q$ in the conditions of Theorem \ref{FirstThofred} such that $F_{L}$ is completely nonholonomic. Then, the following conditions are equivalent:

\begin{itemize}
\item[(i)] $\Gamma_{L,N}^{X}$ and $\Gamma_{L,N} $ are $X$-related.
\item[(ii)] $d \left( E_{L} \circ  X  \right) =0.$
\end{itemize} 
\end{theorem}
\begin{proof}
Due to the fact that $G_{L}$ is completely nonholonomic we can use Proposition \ref{Importantconsequuence2134} to have that $d \left( E_{L} \circ  X  \right) \in G_{L}$ if, and only if, $d \left( E_{L} \circ  X  \right) =0$. Notice that, obviously, $d \left( E_{L} \circ  X  \circ \tau_{Q} \right) \in F_{L}^{o}$ if, and only if, $d \left( E_{L} \circ  X  \circ \tau_{Q} \right) \in G_{L}^{o}$.
\end{proof}

We will see at the end of the section that this theorem generalizes Theorem \ref{nhhj21corollary4}.

In an obvious way, we can present its Hamiltonian counterpart by following the notation introduced for the Theorem \ref{nhhj1}.\\

\begin{theorem}\label{dualtheo23}
Let $\sigma$ be a $1-$form on $Q$ such that $\sigma(Q) \subset M$ and $\left(\sigma \circ \tau_{Q}\right)^{*} \omega_{Q} \in \mathcal{I} \left(F_{L}^{o}\right)$. Then, the following conditions are equivalent:

\begin{itemize}
\item[(i)] $\overline{\Gamma}_{H,M}^{\sigma}$ and $\overline{\Gamma}_{H,M}$ are $\sigma$-related.

\item[(ii)] $d \left( H \circ \left( \sigma \circ \tau_{Q}\right) \right) \in F_{L}^{o}$
\end{itemize} 
\end{theorem}

Here $F^{o}$ is just the pullback of $F^{o}_{L}$ by the inverse of the Legendre transformation $\mathbb{F}L^{-1}$.\\ 

\begin{definition}
A $1-$form $\sigma$ on $Q$ satisfying conditions of the Theorem \ref{dualtheo23} will be called solution for the constrained Hamilton-Jacobi problem given by the constraint manifold $M \subseteq T^{*}Q$ and a Hamiltonian $H$.
\end{definition}

So, Corollary \ref{funnycorollary123} has an obvious counterpart in the Hamiltonian side.
\begin{corollary}\label{dualfunnycorollary123}
Let $\sigma$ be a $1-$form on $Q$ in the conditions of Theorem \ref{dualtheo23} such that $T \sigma \left( TQ \right) \subseteq F$. Then, the following conditions are equivalent:

\begin{itemize}
\item[(i)] $\Gamma_{H,M}^{\sigma}$ and $\Gamma_{H,M} $ are $\sigma$-related.
\item[(ii)] $d \left( H \circ  \sigma  \right) =0.$
\end{itemize} 

\end{corollary}

Again, condition $T \sigma \left( TQ \right) \subseteq F$ is equivanlent to $F = T_{M} \left(T^{*}Q \right)$ and, hence, the equation of motion associated to the constraint Hamiltonian system is defined to be 
\begin{equation}\label{almost32nHELag}
\iota_{\Gamma_{H,M}}\omega_{Q} -dH=0, \qquad
\left. \Gamma_{H,M} \right\vert _{M} \in TM,
\end{equation}
Finally, as a dual version of Theorem \ref{AppChowtheore,} we have that

\begin{theorem}\label{AppChowtheoredual23}
Let $\sigma$ be a $1-$form on $Q$ in the conditions of Theorem \ref{dualtheo23} such that $G_{L}$ is completely nonholonomic. Then, the following conditions are equivalent:

\begin{itemize}
\item[(i)] $\Gamma_{H,M}^{\sigma}$ and $\Gamma_{H,M} $ are $\sigma$-related.
\item[(ii)] $d \left( H \circ  \sigma  \right) =0.$
\end{itemize} 
\end{theorem}

It is interesting to see how the Hamilton-Jacobi theorem for nonholonomic systems with nonlinear constraints, that is Theorem \ref{FirstThofred}, covers the Hamilton-Jacobi theorem for nonholonomic systems with linear constraints, that is Theorem \ref{nhhj21}. See that, for the linear case, 
$$F_{L}^{o} = \tau_{Q}^{*}\left( N_\ell^{o} \right).$$
This gives that the assumption $\left(X \circ \tau_{Q}\right)^{*} \omega_{L} \in \mathcal{I} \left(F^{o}_{L}\right)$ of the Theorem \ref{FirstThofred} reduces to $X ^{*} \omega_{L} \in \mathcal{I} \left(N_\ell^{o}\right).$
A direct computation shows us that this is equivalent to the assumption  $d(\mathbb{F}L \circ X)\in {\mathcal I} (N_\ell^o)$ of the Theorem \ref{nhhj21} . Indeed, by taking into account that $S^{*} \left( dL \right) = \mathbb{F}L^{*} \theta_{Q}$, we compute
\begin{eqnarray*}
X ^{*} \omega_{L} & = &- X^{*} \left( d S^{*} \left( dL \right) \right) = - X^{*} \left( d \mathbb{F}L^{*} \theta_{Q} \right) \\
&=& -d \left( \mathbb{F}L \circ X \right)^{*} \theta_{Q} = - d(\mathbb{F}L \circ X).
\end{eqnarray*}
Analogously, the second condition $d \left( E_{L} \circ \left( X \circ \tau_{Q}\right) \right) \in F^{o}_{L}$ in the Theorem \ref{FirstThofred} reduces to the second condition $d \left( E_{L} \circ  X \right) \in N_\ell^{o}$ in the Theorem \ref{nhhj21}. Thus, Theorem \ref{nhhj21} can be seen as a corollary of Theorem \ref{FirstThofred}. Analogously, Theorem \ref{nhhj21corollary4} is a particular case of Theorem \ref{AppChowtheore,}. The same happens with the dual case.\\

\section{Reduction of nonholonomic systems}


Assume that a Lie group $G$ acts freely and properly on the base manifold $Q$ leaving invariant the Lagrangian $L$ and the constraint manifold $N$ presented in the subsection (\ref{constrLagsec}). 
So, it can be proved (see \cite{CaLeMaDi98}) the following result:

\begin{proposition}
Under the above hypotheses, $\alpha_{L}$, $\omega_{L}$, $E_{L}$, $\mathbb{F}L$, $\Gamma_{L}$ and $F_{L}$ are $G-$invariant, and, therefore, $\Gamma_{L,N}$ is $G-$invariant.
\end{proposition}
The quotient of an object $A$ by $G$ will be denoted by the name of the overlined quantity, i.e., $\overline{A}$. So, let us define the following projections and reduced spaces 
\begin{equation} \label{rhos}
\rho_{Q}: Q \rightarrow \overline{Q}, \qquad \rho_{TQ}: TQ \rightarrow \overline{TQ}, \qquad \rho_{N}:N \rightarrow \overline{N}
\end{equation}
which are the principal bundles associated with the actions of the group $G$ with $Q$, $TQ$ and $N$ respectively.\\
Consider the tangent lift of the projection $\rho_{TQ}$, that is, $T\rho_{TQ} $ a
linear mapping from $TTQ$ to $T\overline{TQ}$ and its kernel is a
subbundle of $TTQ$ consisting of the vertical vectors. We denote the vectical bundle by $\mathcal{V}$. Similarly, we denote the kernel of
the tangent lift $T\rho_{N}$ of the projection $\rho_{N}$ by $ \mathcal{V}_{N}$. Notice that $ \mathcal{V}_{N}$ is, in fact, the restriction of $ \mathcal{V}$ to $TN$.\\
Let us study the projection $\overline{F}^{o}_{L}$ of $F^{o}_{L}$. Notice that, because of the invariance we can define the projection $\overline{S}$ of $S$ through the following diagram,

\begin{equation} \label{Anotherdiagrammoreforthecollection}
  \xymatrix{ T \left( T Q \right)  \ar"1,4"^{S} \ar[dd]^{T\rho_{TQ}}&   & &  T \left( T Q \right) \ar[dd]^{T\rho_{TQ}}\\
  &  & &\\
 T\overline{TQ}\ar"3,4"^{\overline{S}} &  & &   T\overline{TQ}}
\end{equation}
Then, it is an easy computation to show that,

\begin{equation}\label{quotientedequation332}
\overline{F}^{o}_{L} = \overline{S}^{*} \left( T \overline{N}^{o} \right).
\end{equation}
\noindent
Notice that $\overline{F}^{o}_{L}$ is characterized by the following identity,
\begin{equation}\label{Comotienequeser}
\rho_{TQ}^{*} \overline{F}^{o}_{L} = F^{o}_{L}. 
\end{equation}
It is important to remark that, projecting directly Eq. (\ref{nonhollag}) we obtain that
\begin{equation}\label{nonhollagprojected}
\iota_{\overline{\Gamma}_{L,N}}\overline{\omega}_{L}-d\overline{E_L}\in \overline{F}^{o}_{L},\quad \overline{\Gamma}_{L,N}|_{\overline{N}}\in T\overline{N}.
\end{equation}

Recall the symplectic distribution $\mathcal{H}$ on $TQ$ defined in Eq. (\ref{H}). We define a distribution on $TQ$ along $N$ by
\begin{equation}
\mathcal{U}_{v}=\{Y\in \mathcal{H}_{v}:\omega _{L}(Y,Z)(v)=0,\text{ }\forall Z\in \mathcal{V}%
_{v}\cap (F_{L})_v\}.  \label{distU}
\end{equation}%
In other words,
\[
\mathcal{U}=\mathcal{H}\cap (\mathcal{V}\cap F_L)^{\bot }.
\]

One can easily see that the symplectic form $\omega_{L}$ and the $1-$form $dE_{L}$ can be restricted to $\mathcal{U}$ (see Eq. (\ref{naturally})) as the $2-$form $\omega_{\mathcal{U}}$ on $\mathcal{U}$ and the $1-$form $d{E_L}_{\mathcal{U}}$ respectively.\\
$\mathcal{U}$ projects to the distribution $\overline{\mathcal{U}}$ on $\overline{TQ}$ along $\overline{N}$. Notice also that the objects $\omega_{\mathcal{U}}$ and $d{E_L}_{\mathcal{U}}$ are also reducible. We denote the reduced forms by $\omega_{\overline{\mathcal{U}}}$ and $d{\overline{E_L}}_{\overline{\mathcal{U}}}$, respectively. The following theorem is stating the reduced dynamics.
\begin{theorem}\label{theoremBatesSniat24}
Let $\Gamma_{L,N}$ be the solution of the nonholonomic Lagrangian system. Then, $\overline{\Gamma}_{L,N}$ is a section of $\overline{\mathcal{U}}$ satisfying
\begin{equation} \label{RedSymNonhol}
\iota_{\overline{\Gamma}_{L,N}}\omega_{\overline{\mathcal{U}}}=d{\overline{E_L}}_{\overline{\mathcal{U}}}.
\end{equation}
\end{theorem}
It is important to notice that the pair $(\overline{\mathcal{U}},\omega_{\overline{\mathcal{U}}})$ is a symplectic generalized vector bundle on $\overline{N}$, i.e., for each $\overline{v} \in \overline{\mathcal{U}}$ we have that ${\omega_{\overline{\mathcal{U}}}}_{\overline{v}}$ is a non-degenerate two-form on ${\overline{\mathcal{U}}}_{\overline{v}}$. Furthermore, it satisfies the following equality,
$$ d{\overline{E_L}}_{\overline{\mathcal{U}}} = \left( d \overline{E_{L}}_{\vert \overline{N}}\right)_{\overline{\mathcal{U}}}.$$
This reduction was first introduced by L. Bates and J. Sniatycki in \cite{BatesSnia} for the linear case. However, it is a straighforward exercise to prove that Theorem \ref{theoremBatesSniat24} is also satisfied for the non linear case.\\ 
We summarize all the discussion so far proposed in this subsection in the following commutative diagram 
 \begin{equation} \label{barNbarU}
  \xymatrix{ \mathcal{U} \ar[dr]
\ar[dd]_{\tau_{TQ}\vert_{\mathcal{U}}} \ar@{^{(}->}[rrr]&   & &T_NTQ\ar[dr]\ar[dd]_>>>>>>>>>>>>>>>>{\tau_{TQ}}\\
  & \overline{\mathcal{U}}\ar[dd]\ar@{^{(}->}[rrr] & && T_{\overline{N}}\overline{TQ}\ar[dd] \\
 N \ar[dr] \ar@{^{(}->}[rrr]&  & & TQ \ar[dr]
 \\
 & \overline{N}\ar@{^{(}->}[rrr] & & & \overline{TQ}
 }
\end{equation}

For the reduction theory of nonholonomic systems, we may make use of a nonholonomic map that we introduce in the following lines.
For this, we have to make a generalization of the momentum map that we use in the usual mechanics, let us define it.

\subsection{The nonholonomic momentum map}

We present a generalization of the momentum map for nonholonomic systems with symmetry obtained by Bloch et al \cite{BKMM} and extended by M. de Le\'on et al \cite{CaLeMaDi98}. This generalization considers the union of subspaces in $\mathfrak{g}$ in order to account for the constrained dynamics to a submanifold $N$. In the case of the tangent space, for every $q\in Q$ we define the following vector subspace of $\mathfrak{g}$.
\begin{equation}
\mathfrak{g}^q=\{\xi \in \mathfrak{g} \ : \ \xi_{TQ}(v_q)\in  (F_L)_{v_q}, \forall v_q\in T_qQ\cap N\}.
\end{equation}

\noindent

We consider now the disjoint union of these vector spaces, taken at all points of $Q$, i.e.,
 
$$\mathfrak{g}^N=\bigsqcup_{q\in Q} \mathfrak{g}^q \subseteq Q \times \mathfrak{g}$$
\noindent
and let $\pi:\mathfrak{g}^N\rightarrow Q$ be a canonical projection. Notice this is not a vector bundle, because $\mathfrak{g}^N$ is not a differentiable manifold.
We will understand this as a generalized vector bundle, which induces another generalized vector bundle $(\mathfrak{g}^N)^{*}$ defining its fibers as the dual vector spaces of the fibers of $\mathfrak{g}^N$.

So, we define the nonholonomic momentum map as $J^{nh}_L:TQ\rightarrow (\mathfrak{g}^N)^{*}$ by
\begin{equation}
 J^{nh}_L(v_q)\left( \xi \right)=\alpha_L(\xi_{TQ})(v_q),\quad \text{for all}\quad v_q\in TQ,\quad \xi\in \mathfrak{g}^N.
\end{equation}
\noindent

Let $\xi$ be a section of $\mathfrak{g}^N$. Then, we can define a map $J^{nh}_{L,\xi}  : TQ \rightarrow \mathbb{R}$ by putting
$$J^{nh}_{L,\xi} \left( v_{q} \right) =J^{nh}_L(v_q) \left( \xi \left( q \right) \right), \ \forall v_{q} \in TQ.$$
Furthermore, $\xi$ induces a vector field $\Theta$ on $Q$ as follows:

\begin{equation}
\Theta(q)=(\xi(q))_Q(q)
\end{equation}
through which we can rewrite a Noether Theorem \cite{CaLeMaDi98}:

\begin{theorem}
\begin{equation}
\Gamma_{L,N}(J^{nh}_{L,\xi})=\Theta^{c}(L),
\end{equation}
where $\Theta^{c}$ denotes the complete lifting of $\Theta$ onto a vector fiel on $TQ$.
\end{theorem}
If it happens that $\xi$ is a section of $\mathfrak{g}^N\rightarrow Q$ that is constant, i.e., $\xi(q)=\xi \in \mathfrak{g}$ for all $q$, then $\Theta^{c}=\xi_{TQ}$, which in view of the $G$-invariance of $L$, then $\Gamma_{L,N}(J^{nh}_{L,\xi})=0$ and we recover the Noether theorem for the unconstrained nonholonomic momentum map.

\subsection{A classification of constrained systems with symmetry} \label{classification}

In this section, we will present a classification of constrained systems with symmetry. When there exist symmetries, depending on the relative position of the distribution $F_L$ associated with the external forces, and the vertical space associated with the projection $\rho_{TQ}$, we can define three different types of nonholonomic systems. In the following classification, ${(\mathcal{V}_{N})}_{v}$ is the space of vertical vectors at $v$ in $N$ with respect to the projection $\rho_N$ (see (\ref{rhos})), and $\mathcal{H}_{v}$ is the vector space defined in Eq. (\ref{H}).
\begin{enumerate}
\item Pure kinematic case: ${(\mathcal{V}_{N})}_{v}\cap \mathcal{H}_{v}=0$ and $%
T_{v}N={(\mathcal{V}_{N})}_{v}+\mathcal{H}_{v}$ for all $v\in N.$

\item Horizontal symmetries: ${(\mathcal{V}_{N})}_{v}\cap \mathcal{H}_{v}={(\mathcal{V}_{N})}_{v}$ for all $v\in N$ which is equivalent to ${(\mathcal{V}_{N})}_{v}\subset
\mathcal{H}_{v}$ for all $v\in N.$

\item General case: $0\varsubsetneq {(\mathcal{V}_{N})}_{v}\cap\mathcal{H}_{v}\varsubsetneq {(\mathcal{V}_{N})}_{v}$ for all $v\in N.$
\end{enumerate}

Let us examine the pure kinematic case and the horizontal case in the following sections, respectively. For a detailed development of the general case see \cite{CoLe99}.
\subsection{Pure kinematic case} \label{Purekinematic}

As discussed in the previous subsection, in the pure kinematical case, we have two conditions: for each $v \in N$ we have that
\begin{equation}
{\mathcal{V}_{N}}_{v} \cap \mathcal{H}_{v} = \{0\}, \qquad 
T_{v}N = {\mathcal{V}_{N}}_{v} + \mathcal{H}_{v}.
\end{equation}
This reads that we have the following direct sum decomposition of the tangent bundle of the constraint submanifold 
\begin{equation}
TN = \mathcal{V}_{N} \oplus \mathcal{H}.
\end{equation}
It is important to remark that in this case the nonholonomic momentum map is trivial, i.e.,
$$J^{nh} \equiv 0.$$
It is satisfied that $\mathcal{U} = \mathcal{H}$ and, hence, $\mathcal{U}$ is just the horizontal distribution of a principal connection in $\rho_{N}$. This implies that the projected distribution $\overline{\mathcal{U}}$ is equal to the tangent bundle $T \overline{N}$. Therefore, $\overline{\omega}_{L}$ is an almost symplectic two form on $\overline{N}$. To work with a symplectic two form on $\overline{N}$ we define the one form on $N$,
$$\tilde{\alpha} = \iota_{\Gamma_{L,N}}\left( \textbf{h}^{*}d \left( j^{*}\alpha_{L} \right) - d \textbf{h}^{*}\left(j^{*}\alpha_{L} \right)\right),$$
where $j : N \hookrightarrow TQ$ is the inclusion map and $\textbf{h}$ is horizontal projection of $\mathcal{H}$. Here $\tilde{\alpha}$ is $G-$ invariant and horizontal and, hence, it projects to a one form $\overline{\alpha}$ on $\overline{N}$. 

In this situation, the reduced equations of motion are:

\begin{theorem}
The reduced equation of motion can be written as follows,
\begin{equation}\label{121}
\iota_{\overline{\Gamma}_{L,N}}\overline{\omega} = d \overline{E}_{L} - \overline{\alpha},
\end{equation}
where $\overline{\omega} = d \overline{\alpha}_{L}$ is a symplectic two form and $\overline{\alpha}_{L}$ is the projection on $\overline{N}$ of $\textbf{h}^{*}\left(j^{*}\alpha_{L}\right)$. Moreover, we can prove that $\iota_{\overline{\Gamma}_{L,N}}\overline{\alpha} = 0.$
\end{theorem}

\subsection{Horizontal symmetries case} \label{horsymcase}

The case of horizontal symmetries is characterized by the following property:
\begin{equation}
\mathcal{V}_{N} \cap \mathcal{H} = \mathcal{V}_{N}.
\end{equation}
So that $\mathcal{V}_{N}$ becomes a subdistribution of $\mathcal{H}$.
The importance of this case lies on the fact of the nonholonomic momentum map turns into the standard momentum map $J_{L}: TQ \rightarrow \mathfrak{g}^{*}$. In fact,
$$ \mathfrak{g}^{q} = \mathfrak{g}, \ \forall q \in Q.$$
So, by using the Marsden-Weinstein reduction, for each value $\mu \in \mathfrak{g}^{*}$ (all the values are regular) we have a reduced sympletic manifold
$$ \left( \mathcal{P}_{\mu} = J_{L}^{-1} \left( \mu \right) / G_{\mu}, \omega_{L,\mu} \right),$$
where $G_{\mu}$ is the isotropy group of $\mu$ for the co-adjoint action and $\omega_{L,\mu}$ is the unique symplectic $2-$form on $\mathcal{P}_{\mu}$ satisfying
$$ \rho_{L,\mu}^{*} \omega_{L,\mu} = i_{L,\mu}^{*} \omega_{L},$$
with $\rho_{L,\mu}$ is the canonical projection form $J_{L}^{-1} \left( \mu \right)$ to the reduced manifold $\mathcal{P}_{\mu}$ whereas $i_{L,\mu}$ is the inclusion map from $J_{L}^{-1}\left( \mu \right)$ to the tangent bundle $T Q$.\\
Let us consider the projection $N_{\mu} =\left(N / G_{\mu} \right)$ and the canonical projection $\rho_{TQ,\mu} : TQ \rightarrow \left(TQ\right)_{\mu} = TQ/G_{\mu}$. We will assume that $N$ and $J_{L}^{-1} \left( \mu \right)$ have a clean intersection $N^{'} = N \cap J_{L}^{-1} \left( \mu \right)$. So, the $N^{'}$ is obviously $G_{\mu}-$invariant and, hence, it projects onto a submanifold $\mathcal{N}_{\mu}$ of $\mathcal{P}_{\mu}$.\\
Furthermore, the distribution $F_{L}$ induces a distribution $F^{'}_{L}$ on $J_{L}^{-1} \left( \mu \right)$ along $N^{'}$ given by the intersection of $F_{L}$ with the tangent space of $J_{L}^{-1} \left( \mu \right)$ at each vector of $N^{'}$. Again, the $G_{\mu}-$invariance of $F^{'}_{L}$ projects onto a distribution $F_{L, \mu}$ on $\mathcal{P}_{\mu}$ along $\mathcal{N}_{\mu}$.\\
So, let us remember the following Theorem (see \cite{CaLeMaDi98}).
\begin{theorem}
The projection $\left( \Gamma_{L,N} \right)_{\mu}$ of the restriction of $\Gamma_{L,N}$ to $N^{'}$ is a solution of the reduced equations of motion
\begin{equation}\label{136}
\iota_{\left( \Gamma_{L,N} \right)_{\mu}} \omega_{L,\mu} - d \left( E_{L} \right)_{\mu} \in F^{o}_{L,\mu}; \ \ \ \ \left( \Gamma_{L,N} \right)_{\mu} \in T\mathcal{N}_{\mu},
\end{equation} 
where $\left( E_{L}\right)_{\mu}$ is the reduced energy.
\end{theorem}
\noindent Observe that
\begin{equation}\label{130}
\rho^{*}_{TQ, \mu}\left( TN_{\mu}^{o} \right) = {TN}^{o},
\end{equation}
In fact, as a direct consequence of $ T N_{\mu} = T \rho_{TQ,\mu} \left( T N \right)$ we have that,
$$\rho^{*}_{TQ, \mu}\left( TN_{\mu}^{o} \right) \subseteq {TN}^{o}.$$
Then, counting dimensions, we have Eq. (\ref{130}).\\

Let us study the codistribution $F_{L,\mu}^{o}$. Notice that for each $U_{v_{q}} \in T_{v_{q}} \left( TQ \right)$ and $g \in G_{\mu}$ we have

\begin{eqnarray*}
S \left( T_{v_{q}} \phi^{T}_{g} \left( U_{v_{q}} \right) \right) & = & \left( \phi^{T}_{g} \left(v_{q} \right)+ t \left[ T_{v_{q}} \left( \tau_{Q} \circ \phi^{T}_{g}\right) \left( U_{v_{q}}\right)\right] \right)^{'}_{|0}\\
&=& \left( T_{q}\phi_{g} \left(v_{q} \right)+ t \left[ T_{v_{q}} \left( \phi_{g} \circ  \tau_{Q}\right) \left( U_{v_{q}} \right)\right] \right)^{'}_{|0}\\
&=&  \left( T_{q}\phi_{g} \left(v_{q} + t\left[ T_{v_{q}}   \tau_{Q} \left( U_{v_{q}} \right)\right]\right) \right)^{'}_{|0}\\
&=& \left[ T_{v_{q}} \phi^{T}_{g} \circ S \right] \left( U_{v_{q}} \right).
\end{eqnarray*}
Therefore, $S$ is $G_{\mu}-$invariant and it can be projected onto a map
$$S_{\mu} : T \left(TQ\right)_{\mu} \rightarrow T \left(TQ\right)_{\mu},$$
such that $ T \rho_{TQ,\mu} \circ S = S_{\mu} \circ T \rho_{TQ,\mu}$. Then,
\begin{equation}\label{132}
\rho_{TQ,\mu}^{*} S^{*}_{\mu} = S^{*} \rho_{TQ,\mu}^{*}.
\end{equation}
Thus, by definition we have that
$$  \rho_{L,\mu}^{*} F^{o}_{L,\mu} = i_{L,\mu}^{*} \left( S^{*} \left( TN^{o} \right) \right).$$
Then, taking into account Eq. (\ref{130}) and Eq. (\ref{132})
$$  \rho_{L,\mu}^{*} F^{o}_{L,\mu} =\left(\rho_{TQ,\mu} \circ i_{L,\mu}\right)^{*} \left(S_{\mu}^{*} \left( T N_{\mu}^{o} \right) \right) = \left(j_{L,\mu}\circ \rho_{L,\mu}\right)^{*} \left(S_{\mu}^{*} \left( T N_{\mu}^{o} \right) \right), $$
where $j_{L,\mu} :  \mathcal{P}_{\mu} \hookrightarrow \left( TQ\right)_{\mu}$. Then,
\begin{equation}\label{133}
F^{o}_{L,\mu} = j_{L,\mu}^{*}\left(S_{\mu}^{*} \left( TN_{\mu}^{o} \right)\right).
\end{equation}
Let us consider the constraints $\psi^{a}$ which characterize (locally) $N$. Then, by using Eq.(\ref{133}), there exists a family of (local) $1-$forms on $\left(TQ \right)_{\mu}$
\begin{equation}\label{aux34d}
 j_{L,\mu}^{*} S_{\mu}^{*} \left(  \overline{d\psi^{a}} \right),
\end{equation}
which generates $F^{o}_{L,\mu}$.\\

\section{Hamilton--Jacobi theory for reduced nonholonomic Lagrangian systems}

In this section we will present our second goal: The reduction of the Hamilton--Jacobi equation for nonholonomic systems. For this, let us take into consideration all the aforementioned notation and theory. In the following subsection, we start with a reduction of a Hamilton--Jacobi theory for nonholonomic systems in the more general case, and then consider the kinematic and horizontal cases as particular cases of this general picture.


\subsection{General case}\label{Generalcase3}

Let us consider the symplectic manifold $(TQ,\omega_L)$. This implies that the pair $(T_vTQ,\omega_L(v))$ is a symplectic vector space at each point $v$ in $TQ$. In other words, the triple $(TTQ,TQ,\omega_L)$ is a symplectic vector bundle \cite{we}.

Let $\overline{X}$ be a section of the projection $\overline{\tau}_{Q}: \overline{TQ} \rightarrow \overline{Q}$ such that $\overline{X} \left( \overline{Q} \right) \subseteq \overline{N}$. Then, we may define a vector field $\overline{\Gamma}^{\overline{X}}_{L,N}$ on $\overline{Q}$ by
\begin{equation}\label{commdefi3e4d}
\overline{\Gamma}^{\overline{X}}_{L,N} = T \overline{\tau}_Q \circ \overline{\Gamma}_{L,N} \circ \overline{X}.
\end{equation} 
Hence, the following diagram is commutative
\begin{equation} \label{hattau2}
  \xymatrix{ \overline{N} \ar"1,4"^{\overline{\Gamma}_{L,N}}&   & & T \left( \overline{TQ} \right)  \ar[dd]^{T\overline{\tau}_Q}\\
  &  & &\\
 \overline{Q} \ar"3,4"^{\overline{\Gamma}^{\overline{X}}_{L,N}} \ar"1,1"^{\overline{X}}&  & & T\overline{Q} }
\end{equation}

\noindent
Let us study carefully the section $\overline{X}$. Notice that, taking into account that the action on $Q$ is free, for each $q \in Q$ there exists just one tangent vector $X_{q} \in T_{q} Q$ such that
$$ \overline{X} \left( \overline{q} \right) = \rho_{TQ} \left( X_{q}\right).$$
Thus, we can define the map $X: Q \rightarrow TQ$ by the equation: $X \left( q \right) = X_{q}$ for all $q \in Q$. By uniqueness, the map $X$ is $G-$invariant.\\
Now, to prove that $X$ is a \textit{bona-fide} vector field we should show that it is differentiable.\\
Let us consider two (local) sections $\Theta$ and $\theta$ of $\rho_{TQ}$ and $\rho_{Q}$ respectively such that the following diagram is commutative:
\begin{equation} \label{sectionslocal441}
  \xymatrix{ \overline{TQ} \ar"1,4"^{\Theta} \ar[dd]^{\overline{\tau}_Q}&   & &  TQ   \ar[dd]^{\tau_Q}\\
  &  & &\\
 \overline{Q} \ar"3,4"^{\theta} &  & & Q }
\end{equation}
On the other hand we know that,
\begin{equation}\label{Ginvaiantvector2323}
\rho_{TQ} \circ X = \overline{X} \circ \rho_{Q}.
\end{equation}
So, $X^{\Theta}= \Theta \circ \rho_{TQ} \circ X$ is differentiable. However, $X^{\Theta}$ is not a vector field and, obviously, 
$$ X^{\Theta} \neq X.$$
In fact, $X^{\Theta}$ is constant on the orbits of the action of $G$ over $Q$. By using that $\theta$ is a section, $\theta \left( \overline{Q} \right)$ is a submanifold of $Q$. Hence, the restriction $X^{\theta}$ of $X^{\Theta}$ to $\theta \left( \overline{Q} \right)$ is differentiable. In fact, by uniqueness
$$ X^{\theta} = X_{\vert \theta \left( \overline{Q}\right)}.$$
So, by using the multiplication of the action, we finally prove that $X$ is differentiable and, hence, $X$ is a $G-$invariant vector field on $Q$ such that it satisfies Eq. (\ref{Ginvaiantvector2323}). Equivalently, we have proved that there exists a one-to-one correspondence 
\begin{equation}\label{reconstrucion23}
\begin{array}{rccl}
&{\mathfrak{X}}_{G} \left( Q \right) & \rightarrow & \Gamma \left( \overline{\tau}_{Q} \right) \\
& X  &\mapsto &  \overline{X}
\end{array}
\end{equation}
where $\mathfrak{X}_{G} \left(Q \right)$ is the space of $G-$invariant vector fields on $Q$ and $\Gamma \left(\overline{\tau}_{Q}  \right)$ is the space of sections of $\rho_{TQ}$. We will use this correspondence along the rest of the paper.\\
Next, we want to prove a reduced version of the Theorem \ref{FirstThofred}. To do this, we will define the ideal ${\mathcal I}(\overline{F}_{L}^o)$ in the same way of ${\mathcal I}(F_{L}^o)$ (see Eq. ( \ref{definitiononeoftheamountofideals23})).

We now state the reduced theorem:

\begin{theorem}\label{SecondThofred}
Let $\overline{X}$ be a section of $\overline{\tau}_{Q}$ such that $\overline{X} \left( \overline{Q} \right) \subset \overline{N}$ and $\left(\overline{X} \circ \overline{\tau}_{Q}\right)^{*} \overline{\omega}_{L} \in \mathcal{I} \left(\overline{F}^{o}_{L}\right)$. Then, the following conditions are equivalent:

\begin{itemize}
\item[(i)] $\overline{\Gamma}_{L,N}^{\overline{X}}$ and $\overline{\Gamma}_{L,N} $ are $\overline{X}$-related.
\item[(ii)] $d \left( \overline{E}_{L} \circ \left( \overline{X} \circ \overline{\tau}_{Q}\right) \right) \in \overline{F}^{o}_{L}$
\end{itemize} 
\end{theorem}
\begin{proof}
We will use the correspondence between $G-$invariant vector fields on $Q$ and sections of $\overline{\tau}_{Q}$, $X \mapsto \overline{X}$, and Theorem \ref{FirstThofred} to prove this. So, let $X$ be the $G-$invariant vector field on $Q$ which projects on $\overline{X}$. Then,
$$ \overline{X} \left( \overline{Q} \right) = \rho_{TQ} \left( X \left( Q \right) \right).$$
So, condition $\overline{X} \left( \overline{Q} \right) \subset \overline{N}$ turns into the following condition,
$$\rho_{TQ}\left(X \left( Q \right) \right)\subset \rho_{TQ} \left( N \right).$$
By the $G-$invariance of $N$ we have that this is equivalent to
$$X \left( Q \right) \subset N.$$
Next, let us work with condition $\left(\overline{X} \circ \overline{\tau}_{Q}\right)^{*} \overline{\omega}_{L} \in \mathcal{I} \left(\overline{F}^{o}_{L}\right)$, i.e.,

$$\left(\overline{X} \circ \overline{\tau}_{Q}\right)^{*} \overline{\omega}_{L} = \overline{\beta}_a\wedge \overline{S}^{*} \left( \overline{d \psi^{a}} \right),$$
with $\overline{\beta}_{a} \in \Lambda^1 \left( \overline{TQ} \right)$. Then, by taking into that $\rho_{TQ}$ is a submersion, this easily implies that, $\left(X \circ \tau_{Q}\right)^{*} \omega_{L} \in \mathcal{I} \left(F^{o}_{L}\right)$. So, we are in the conditions of Theorem \ref{FirstThofred}. Finally, it is just an exercise to prove that conditions $(i)$ and $(ii)$ are equivalent to conditions $(i)$ and $(ii)$ of the Theorem \ref{FirstThofred}.

\end{proof}

\begin{definition}
A section $\overline{X}$ of $\overline{\tau}_{Q}$ satisfying conditions of the Theorem \ref{SecondThofred} will be called solution for the reduced constrained Hamilton-Jacobi problem given by the constraint manifold $N \subseteq TQ$, the regular Lagrangian $L$ and the action of the group $G$ on $Q$.
\end{definition}

By using the proof of the theorem we actually have a result giving the relation between $G-$invariant solutions of the Hamilton-Jacobi problem and solutions of the reduced constrained Hamilton-Jacobi problem. 
\begin{proposition}
Let $X$ be a $G-$invariant vector field on $Q$. Then, $X$ is a solution for the constrained Hamilton-Jacobi problem if, and only if, $\overline{X}$ is a solution for the reduced constrained Hamilton-Jacobi problem.
\end{proposition}

Notice that at the beginning of this subsection we have given a constructive way to obtain the associated $G-$invariant vector field $X$ on $Q$ to a section $\overline{X}$ of $\overline{\tau}_{Q}$. In fact, to have the sections $\theta$ and $\Theta$ in the diagram (\ref{sectionslocal441}) we only need to use a principal connection on $\rho_{Q}$. So, this process defines a detailed way of reconstruction from solutions for the reduced constrained Hamilton-Jacobi problem to solutions for the constrained Hamilton-Jacobi problem.

\begin{corollary} \label{corollary-}
Let $\overline{X}$ be a section of $\overline{\tau}_{Q}$ in the conditions of the Theorem \ref{SecondThofred} satisfying $T\overline{X} \left( T \overline{Q} \right) \subseteq \overline{F}_{L}$.
Then, the following conditions are equivalent:
\begin{itemize}
\item[(i)] $\overline{\Gamma}^{\overline{X}}_{L,N}$ and $\overline{\Gamma}_{L,N}$ are $\overline{X}$-related.
\item[(ii)] $d \left( \overline{E}_{L} \circ \overline{X} \right) = 0$
\end{itemize}
\end{corollary}

This result is just a reduced version of the Corollary \ref{funnycorollary123}. Condition $T\overline{X} \left( T \overline{Q} \right) \subseteq \overline{F}_{L}$ is actually 
$$T \rho_{TQ} \left(TX \left( T Q \right) \right)\subseteq \overline{F}_{L}.$$
However, the associated $G-$invariant vector field $X$ does not have to satisfy conditions of Corollary \ref{funnycorollary123} (see Example \ref{firstexample324}). So, it turns out curious how in the reduced version we obtain new conditions which are not so restrictive (Remember that conditions of Corollary \ref{funnycorollary123} are only possible in a contraint system in which the constraints do not produce reaction forces).\\
Finally, let us give a reduced version of Theorem \ref{AppChowtheore,}. Notice that the distribution $G_{L}$ is $G-$invariant. In fact, we only have to prove that the set of (local) $1-$forms $\mathcal{S}$ is invariant by the action.\\
Thus, we may defined the quotiented distribution $\overline{G}_{L}$ on $\overline{TQ}$. Obviously $\overline{G}_{L}$ restricts to $\overline{F}_{L}$ along $\overline{N}$. Observe that $\overline{G}_{L}$ could be alternatively defined following the same way of $G_{L}$ by using Eq. (\ref{quotientedequation332}).\\
Again we will define the completely nonholonomic condition on $\overline{F}_{L}$. So, consider the family of (local) vector fields $\overline{\mathcal{S}}_{\overline{N}}^{\sharp}$ on $\overline{TQ}$ tangents to $\overline{G}_{L}$ whose domains have non-empty intersection with $\overline{N}$. Then, $\overline{F}_{L}$ can be seen as the distribution generated by $\overline{\mathcal{S}}_{\overline{N}}^{\sharp}$ along $\overline{N}$.\\
Thus, $\overline{F}_{L}$ is called completely nonholonomic if the reiterated Lie bracket of vector fields in $\overline{\mathcal{S}}_{\overline{N}}^{\sharp}$ generate $T_{\overline{N}} \left( \overline{TQ} \right)$. In this case we will say that $\left[\overline{F}_{L},\overline{F}_{L}\right], \left[ \overline{F}_{L}, \left[ \overline{F}_{L},\overline{F}_{L} \right] \right], \hdots$ spans the vector bundle $T_{\overline{N}}\left( \overline{TQ} \right)$.
Obviously, $\overline{G}_{L}$ is completely nonholonomic if, and only if, $\overline{F}_{L}$ is completely nonholonomic.\\

So, in an analogous way to Theorem \ref{AppChowtheore,} we have the following result
\begin{theorem}\label{AppChowtheoreprojected22,}
Let $\overline{X}$ be a section of $\overline{\tau}_{Q}$ in the conditions of Theorem \ref{SecondThofred} such that $\overline{F}_{L}$ is completely nonholonomic. Then, the following conditions are equivalent:

\begin{itemize}
\item[(i)] $\overline{\Gamma}^{\overline{X}}_{L,N}$ and $\overline{\Gamma}_{L,N}$ are $\overline{X}$-related.
\item[(ii)] $d \left( \overline{E}_{L} \circ \overline{X} \right) = 0$
\end{itemize} 
\end{theorem}

Let us give a result relating the distribution $F_{L}$ with its projected one.
\begin{proposition}\label{Relation23}
$F_{L}$ is completely nonholonomic if, and only if, $\overline{F}_{L}$ is completely nonintegrable.
\end{proposition}
\begin{proof}
Remember that $F_{L}$ (resp. $\overline{F}_{L}$) is generated by the family of local vector fields $\mathcal{S}_{N}^{\sharp}$ (resp. $\overline{\mathcal{S}}_{\overline{N}}^{\sharp}$) restricted to $N$ (resp. $\overline{N}$). Furthermore, by taking into account (see Eq. \ref{Comotienequeser}) that
$$\overline{F}_{L} = T \rho_{TQ} \left(F_{L}\right)$$
we can find a family of $G-$invariant vector fields which generates $\mathcal{S}_{N}^{\sharp}$.\\

Now, suppose first that $\overline{F}_{L}$ is completely nonintegrable. Let $V_{v_{q}}$ be a vector in $T_{v_{q}} \left( TQ \right)$. Fix a (local) $G-$invariant vector field $X$ on $TQ$ such that
$$ X \left( v_{q} \right) = V_{v_{q}}.$$
Then, the projection $\overline{X}$ of $X$ can be obtained by a linear combination of Lie brackets of (local) vector fields $\overline{X}^{a}$ on $\overline{TQ}$ tangents to $\overline{F}_{L}$. Hence, by using the associated vector fields $X^{a}$ of $\overline{X}^{a}$ we write $X \left( v_{q} \right)$ as a linear combination of Lie brackets of (local) vector fields tangents to $F_{L}$. So, $F_{L}$ is completely nonintegrable. Notice that the projection of $G-$invariant vector fields on $TQ$ into sections of $\overline{\tau}_{Q}$ preserves the Lie bracket.\\

Conversely, assume that $F_{L}$ is completely nonintegrable. Suppose that we hace three (local) $G-$invariant vector fields $X,Y,Z \in \mathfrak{X} \left( TQ \right)- \{0\}$ and a function $f \in \mathcal{C}^{\infty} \left( TQ \right)$ such that
$$\left[ X , f Y \right] = Z.$$
Then, for all $v_{q}$ in the domain of the vector fields and for all $g \in G$,s
$$ X \left( v_{q} \right) \left( f \circ \phi^{T}_{q} \right) = X \left( v_{q} \right) \left( f \right),$$
where $\phi$ is the action of the group $G$ and $\phi^{T}$ is its tangent lift.\\
By using this fact and taking into account that $\mathcal{S}_{N}^{\sharp}$ is generated by (local)  $G-$invariant vector fields on $TQ$ we can prove that any tangent vector in $T_{N} \left( TQ \right)$ can be written as a linear combination of Lie brackets of (local) $G-$invariant vector fields. Then, projecting these vector fields we have the result.

\end{proof}

Notice that, invariance of the Lagrangian function enables us to transfer this reduced Hamilton-Jacobi theory to the cotangent bundle. To do this, first recall the Legendre transformation $\mathbb{F}L$ given in (\ref{LegTrf}). As discussed before, $\mathbb{F}L$ is a (local) diffeomorphism if the Lagrangian is regular. Using this diffeomorphism we define a submanifold $M$, and a Hamiltonian function $H$ on $T^*Q$ as follows
\begin{equation}
M= \mathbb{F}L \left( N \right),
\qquad
H = E_{L} \circ \mathbb{F}L^{-1}.
\end{equation}
Notice that, if the Legendre transform is not a global diffeomorphism, the Hamiltonian $H$ is only defined on a local neighbourhood.\\
Nex, by using the cotangent lift of the action of $G$ on $Q$, $H$ and $M$ are also invariant under group operation. Then, the projections $\overline{M}$ and $\overline{H}$ on $\overline{T^{*}Q}$ are given by $$\overline{M}=\overline{\mathbb{F}L} \left( \overline{N} \right), \qquad \overline{H}=\overline{E}_{L} \circ \overline{\mathbb{F}L}^{-1},$$ respectively. Here, $\overline{\mathbb{F}L}$ is the projection of $\mathbb{F}L$ from the reduced tangent bundle $\overline{TQ}$ to the reduced cotangent bundle $\overline{T^*Q}$. So, Let us fix the following notation
\begin{equation}
\overline{\sigma} = \overline{\mathbb{F}L}\circ \overline{X},
\quad
\overline{\Gamma}_{H,M} = {\overline{\mathbb{F}L}}_{*} \overline{\Gamma}_{L,N}.
\end{equation}
Now, using the section $\overline{\sigma}$ of $\overline{\pi}_{Q} : \overline{T^{*}Q} \rightarrow \overline{Q}$ such that $\overline{\sigma} \left( \overline{Q} \right) \subseteq \overline{M}$, we define a vector field induces a vector field $\overline{\Gamma}_{H,M}^{\overline{\sigma}}$ on $\overline{Q}$ as follows
\begin{equation}
\overline{\Gamma}_{H,M}^{\overline{\sigma}}=T\overline{\pi}_{Q} \circ \overline{\Gamma}_{H,M}\circ \overline{\sigma}.
\end{equation}
We can now state the following theorem.

\begin{theorem}\label{ThirdThofred}
Let $\overline{\sigma}$ be a $1-$form on $Q$ such that $\overline{\sigma} \left( \overline{Q} \right) \subset \overline{M}$ and $\left(\overline{\sigma} \circ \overline{\tau}_{Q}\right)^{*} \overline{\omega}_{Q} \in \mathcal{I} \left(\overline{F}^{o}_{L}\right)$. Then, the following conditions are equivalent:

\begin{itemize}
\item[(i)] $\overline{\Gamma}_{H,M}^{\overline{\sigma}}$ and $\overline{\Gamma}_{H,M} $ are $\overline{\sigma}$-related.
\item[(ii)] $d \left( \overline{H} \circ \left( \overline{\sigma} \circ \overline{\tau}_{Q}\right) \right) \in \overline{F}^{o}_{L}$
\end{itemize} 
\end{theorem}

\begin{definition}
A section $\overline{\sigma}$ of $\overline{\pi}_{Q}$ satisfying conditions of the Theorem \ref{ThirdThofred} will be called solution for the reduced constrained Hamilton-Jacobi problem given by the constraint manifold $M \subseteq T^{*}Q$, the Hamiltonian $H$ and the action of the group $G$ on $Q$.
\end{definition}

Again, the space of $G-$invariant $1-$forms $\sigma$ on $Q$ are in one-to-one correspondence to sections $\overline{\sigma}$ of $\overline{\pi}_{Q}$
\begin{proposition}
Let $\sigma$ be a $G-$invariant $1-$form on $Q$. Then, $\sigma$ is a solution for the constrained Hamilton-Jacobi problem if, and only if, $\overline{\sigma}$ is a solution for the reduced constrained Hamilton-Jacobi problem.
\end{proposition}

\begin{corollary}
Let $\overline{\sigma}$ be a section of $\overline{\pi}_{Q}$ in the condition of Theorem \ref{ThirdThofred} satisfying $T\overline{\sigma} \left( T \overline{Q} \right) \subseteq \overline{F}$.
Then, the following conditions are equivalent:
\begin{itemize}
\item[(i)] $\overline{\Gamma}^{\overline{\sigma}}_{H,M}$ and $\overline{\Gamma}_{H,M}$ are $\overline{\sigma}$-related.
\item[(ii)] $d \left( \overline{H}\circ \overline{\sigma} \right) = 0$
\end{itemize}
\end{corollary}
\noindent
Here, $F$ is the transformed distribution of $F_{L}$ on $T^{*}Q$ via $\mathbb{F}L$ and $\overline{F}$ is the projection on $\overline{T^{*}Q}$ of $F$.\\
Finally, we can give a Hamiltonian version of the Theorem \ref{AppChowtheoreprojected22,}.

\begin{theorem}\label{AppChowtheoreprojected22dualversionn}
Let $\overline{\sigma}$ be a section of $\overline{\pi}_{Q}$ in the condition of Theorem \ref{ThirdThofred} such that $\overline{F}_{L}$ is completely nonholonomic. Then, the following conditions are equivalent:
\begin{itemize}
\item[(i)] $\overline{\Gamma}^{\overline{\sigma}}_{H,M}$ and $\overline{\Gamma}_{H,M}$ are $\overline{\sigma}$-related.
\item[(ii)] $d \left( \overline{H}\circ \overline{\sigma} \right) = 0$
\end{itemize} 
\end{theorem}

Now, we will study two particular cases: the {purely kinematic case} and the {horizontal case}.
\subsection{Pure kinematic case}

In section (\ref{Purekinematic}), the reduction of the nonholonomic system (in the pure kinematic case) has been established and a symplectic equation has been exhibited in Theorem (\ref{121}). Referring to this relation, we can understand the reduced nonholonomic dynamics as a Hamiltonian system with an external force $-\overline{\alpha}$. Let us recall this equation here once more to be more explicit. 
\begin{equation} \label{redPureKinetamic}
\iota _{\overline{\Gamma}_{L,N}}\overline{\omega}-d\overline{E}_{L}=-\overline{\alpha}.
\end{equation}
The external force $-\overline{\alpha}$ is semibasic since it is derived from semibasic Poincar\'e-Cartan one-form $\alpha_L$ by means of natural operators, such as projections and pull-backs.\\

Let $\overline{X}$ be a section of the tangent bundle of $\overline{\tau}_{Q}$ such that $\overline{X} \left( \overline{Q} \right)$ takes values in $\overline{N}$. Notice that $\overline{\Gamma}_{L,N}$ is a vector field on the reduced manifold $\overline{N}$, so that we define a vector field $\overline{\Gamma}_{L,N}^{\overline{X}}$ on basis manifold $\overline{Q}$ following identity (\ref{commdefi3e4d}).
In accordance with this, we will give a reduced version of the Hamilton-Jacobi Theorem \ref{HJforConst} to the system (\ref{redPureKinetamic}).

\begin{theorem}\label{purkincase324d}
Let $\overline{X}$ be a section of $\overline{\tau}_{Q}$ such that $\overline{X} \left( \overline{Q} \right) \subseteq \overline{M}$ and $\overline{X}^{*} \overline{\omega}=0$. Then, the following two conditions are equivalent:
\begin{itemize}
\item [(i)] $\overline{\Gamma}_{L,N}^{\overline{X}}$ and $\overline{\Gamma}_{L,N}$ are $\overline{X}$-related.
\item [(ii)] $d(\overline{E}_{L}\circ \overline{X})=-\overline{X}^*\overline{\alpha}$.
\end{itemize}

\begin{proof}
By using Eq. (\ref{redPureKinetamic}) we have that
$$ \overline{X}^{*}\left( \iota_{\overline{\Gamma}_{L,N}}\overline{\omega} \right) = d \left( \overline{E}_{L} \circ \overline{X} \right) - \overline{X}^{*}\overline{\alpha}.$$
So, condition $(ii)$ can be rewritten as follows
\begin{equation}\label{125}
\overline{X}^{*}\left( \iota_{\overline{\Gamma}_{L,N}}\overline{\omega} \right) =0,
\end{equation}
i.e., $\overline{X} \left( \overline{Q} \right)$ is an isotropic submanifold of $\overline{N}$. On the other hand, in the purely kinematic case, the dimension of $G$ is equal to the maximum number of independent constraints. Then,
$$  dim \left( \overline{N} \right) = 2 \left( dim \left( Q \right) - dim \left( G \right) \right) = \dfrac{1}{2} \left( \overline{X} \left( \overline{Q} \right)\right).$$
Therefore, $\overline{X} \left( \overline{Q} \right)$ is a Lagrangian submanifold of $\overline{N}$. Now, suppose that $(i)$ holds. Then, by Eq. (\ref{redPureKinetamic}) we have that
$$
\overline{X}^{*}\left( \iota_{\overline{\Gamma}_{L,N}}\overline{\omega} \right)=\iota_{ \overline{\Gamma}^{\overline{X}}_{L,N}}\left(\overline{X}^{*}\overline{\omega} \right) =0.$$
Conversely, we only need to use the fact of that $\overline{X} \left( \overline{Q} \right)$ is a Lagrangian submanifold of $\overline{N}$ to prove $(ii)$.
\end{proof}

\end{theorem}
\noindent

An interesting but very particular example is given by a section $\overline{X}$ of $\overline{\tau}_{Q}$ which satisfies conditions of Theorem \ref{purkincase324d} and Theorem \ref{AppChowtheore,}. Then, $\overline{\Gamma}_{L,N}^{\overline{X}}$ and $\overline{\Gamma}_{L,N}$ are $\overline{X}$-related if, and only if,
$$\overline{X}^*\overline{\alpha} = 0.$$

Now, we will transform this reduced picture to the cotangent bundle. Let $\overline{\sigma}$ be the section of $\overline{\pi}_{Q}$ given by
$$\overline{\sigma} = \overline{\mathbb{F}L}\circ \overline{X}.$$
and define,
\begin{itemize}
\item[] $\overline{\theta} = {\overline{\mathbb{F}L}^{-1}}^{*} \overline{\alpha}.$

\item[] $\overline{\Omega }= {\overline{\mathbb{F}L}^{-1}}^{*} \overline{\omega}.$
\end{itemize}
Then, in the reduced picture we have
$$ \overline{\sigma}^{*} \overline{\theta} = \overline{X}^{*}\overline{\alpha}.$$
Furthermore, the vector field $\overline{\Gamma}_{H,M}$ is the solution for the reduced equation,
\begin{equation}\label{1241}
\iota_{\overline{\Gamma}_{H,M}}\overline{\Omega} = d \overline{H} - \overline{\theta},
\end{equation}
So, we can now state the following equivalent theorem.

\begin{theorem}
Let $\overline{\sigma}$ be a section of $\overline{\pi}_{Q}$ such that $\overline{\sigma} \left( \overline{Q} \right) \subseteq \overline{M}$ and $\overline{\sigma}^{*} \overline{\Omega}  = 0$. Then, the following conditions are equivalent:

\begin{itemize}
\item[(i)] $\overline{\Gamma}_{H,M}^{\overline{\sigma}}$ and $\overline{\Gamma}_{H,M}$ are $\overline{\sigma}$-related
\item[(ii)] $d \left( \overline{H} \circ \overline{\sigma} \right) = \overline{\sigma}^{*}\overline{\theta}$
\end{itemize} 
\end{theorem}

\subsubsection*{$N_{\ell}$ is a linear subbundle of $TQ$}

As a particular case, let us consider the case in which $N_{\ell}$ is a vector subbundle of $TQ$ and the constraint functions $\psi^{a}$ are linear. We will follow the notations introduced in subsection \ref{Purekinematic}.\\
Let $X$ be a $G-$invariant vector field on $Q$ such that $X\left( Q \right) \subseteq N_{\ell}$ and $\mathbb{F}L \circ X$ is a closed one form. Then, there exists a unique section $\overline{X}$ of $ \overline{\tau}_{Q}$ which satisfies that
$$ \rho_{TQ} \circ X = \overline{X} \circ \rho_{Q}.$$
Hence, $\overline{X}$ satisfies that $\overline{X}\left( \overline{Q} \right) \subseteq \overline{N}_{\ell}$. On the other hand, by the condition $X\left( Q \right) \subseteq N_{\ell}$ we have that for all $v_{q} \in N_{\ell} \cap T_{q}Q$
$$ T_{q}X \left(v_{q} \right) \in T_{X \left( q \right)} N_{\ell}.$$
Then, taking into account the linearity of $N_{\ell}$, for each $a$
$$ S^{*} \left( d {\psi^{a}}_{|X \left( q \right)} \right) \left(  T_{q}X \left(v_{q} \right)\right) = \left(\psi^{a} \left( X \left( q \right) + t v_{q}\right)\right)' = 0.$$
i.e., for all $v_{q} \in N \cap T_{q}Q$
$$ T_{q}X \left(v_{q} \right) \in (F_L)_{X \left( q \right)}.$$
Equivalently, 
\begin{equation}\label{126}
T_{q}X \left(v_{q} \right) \in \mathcal{H}_{X \left( q \right)}, \ \forall v_{q} \in N_{\ell} \cap T_{q}Q
\end{equation}
Therefore, $\textbf{h} \left( T_{q}X \left(v_{q} \right)\right) = T_{q}X \left(v_{q} \right)$.\\
With this, for all $v_{q} \in T_{q}Q$ 
\begin{eqnarray*}
\left[X^{*}\textbf{h}^{*}\left(j^{*}\alpha_{L}\right)\left( q \right) \right]\left( v_{q}\right) &=& \alpha_{L} \left( X\left( q \right) \right) \left(T_{X \left( q \right) } j \left(\textbf{h} \left( T_{q} X \left( v_{q} \right) \right)\right)\right)\\\\
&=& \alpha_{L} \left( X\left( q \right) \right) \left(T_{X \left( q \right) } j \left( T_{q} X \left( v_{q}  \right)\right)\right)\\\\
&=& \alpha_{L} \left( X\left( q \right) \right) \left(T_{q} X \left( v_{q} \right)\right)\\\\
&=& \left[X^{*}\left(\alpha_{L} \right)\left(  q  \right)\right] \left( v_{q} \right)
\end{eqnarray*}
i.e., 
\begin{equation}\label{127}
X^{*}\textbf{h}^{*}\left(j^{*}\alpha_{L}\right) = X^{*}\left(\alpha_{L} \right).
\end{equation}
Notice that, as an abuse of notation, we have been denoting $X$ and $j \circ X$ (i.e., the restriction of the codomain to $N$) equal.\\
Then, again with some abuse of notation, we have that
\begin{eqnarray*}
\overline{X}^{*} \overline{\omega} &=& d \overline{X}^{*}\overline{\alpha}_{L}\\\\
&=& d \overline{X^{*}\textbf{h}^{*}\left(j^{*}\alpha_{L}\right)}\\\\
&=& d \overline{X^{*}\alpha_{L}}.
\end{eqnarray*}
On the other hand, we know that $\mathbb{F}L  \circ X$ is a closed one form is equivalent to $X^{*} \omega_{L}=d X^{*}\alpha_{L}=0,$ and hence, $\overline{X}^{*} \overline{\omega}=0.$
By using this fact we can give the following two equivalent corollaries:

\begin{corollary}
Let $X$ be a $G-$invariant vector field on $Q$, and $X$ a section such that $X\left( Q \right) \subseteq N$ and $\mathbb{F}L \circ X$ is a closed one form $\sigma$. Then, the following conditions are equivalent:

\begin{itemize}
\item[(i)]  $\overline{\Gamma}_{L,N}^{\overline{X}}$ and $\overline{\Gamma}_{L,N}$ are $\overline{X}$-related.

\item[(ii)] $d \left( \overline{E}_{L} \circ \overline{X} \right) = \overline{X}^{*}\overline{\alpha}$
\end{itemize} 
\end{corollary}

\begin{corollary}
Let $\sigma$ be a $G-$invariant closed one form on $Q$. Then, the following conditions are equivalent:

\begin{itemize}
\item[(i)] $\overline{\Gamma}_{LH,M}^{\overline{\sigma}}$ and $\overline{\Gamma}_{H,M}$ are $\overline{\sigma}$-related.

\item[(ii)] $d \left( \overline{H} \circ \overline{\sigma} \right) = \overline{\sigma}^{*}\overline{\theta}$
\end{itemize}
\end{corollary} 

Notice that we are projecting $G-$invariant one forms $\sigma$ on $Q$ to get one forms $\overline{\sigma}$ on $\overline{Q}$ as we made with the vector fields.\\

Finally, we consider the \textit{\u{C}aplygin systems}. We will say that our nonholonomic system is a \textit{\u{C}aplygin system} when the constraint manifold $N$ is given by the total space of an horizontal distribution $\mathcal{C}$ of the principal connection $\rho_{Q}$. Then,
$$ TQ = \mathcal{C} \oplus \mathcal{V}_{Q},$$
where $\mathcal{V}_{Q}$ is the vertical distribution of $\rho_{Q}$. Its associated horizontal lift $\bf{h_{\mathcal{C}}}$ is a just section of $T \rho_{Q}: T Q \rightarrow T \overline{Q}$ that satisfies
$${\bf{ h_{\mathcal{C}}}} \left( v_{\overline{q}} \right) \in \mathcal{C}_{q}, \ \forall v_{\overline{q}} \in T_{\overline{q}} \overline{Q}.$$
Then, the restriction of $T \rho_{Q} $ to $\mathcal{C}$ is an isomorphism and, therefore,
$$ \mathcal{C} \cong \overline{\mathcal{C}} \cong T \overline{Q}.$$
So, we may define a Lagrangian function $L^{*}: T \overline{Q} \rightarrow \mathbb{R}$ such that
$$ L^{*} \left( v_{\overline{q}} \right) = L \left({\bf{ h_{\mathcal{C}}}} \left( v_{\overline{q} }\right) \right), \ \forall v_{\overline{q}} \in T_{\overline{q}}\overline{Q}.$$

A computation shows that
$$ \overline{\omega} = \omega_{L^{*}}, \ \ \ \ \  \ \ \ \ \overline{E}_{L} = E_{L^{*}}.$$
Thus, the reduced equation of motion is
\begin{equation}\label{128}
\iota_{\overline{\Gamma}_{L,N}}\omega_{L^{*}} = d E_{L^{*}} - \overline{\alpha},
\end{equation}
Then, by an slight abuse of notation, in this case we can given the following results. The proof can be found in \cite{LeonDMDD}.
\begin{theorem} \label{HJ-red-Cap}
Let $\overline{X}$ be vector field of $\overline{Q}$ such that $\mathbb{F}L^{*} \circ \overline{X}$ is closed. Then, the following conditions are equivalent:

\begin{itemize}
\item[(i)] $\overline{\Gamma}_{L,N}^{\overline{X}}$ and $\overline{\Gamma}_{L,N}$ are $\overline{X}$-related.

\item[(ii)] $d \left( E_{L^{*}} \circ \overline{X} \right) = \overline{X}^{*}\overline{\alpha}$
\end{itemize} 
\end{theorem}
Notice that to consider, for instance, $\overline{X}^{*}\overline{\alpha}$ we should see $\alpha$ as a form on $T \overline{Q}$.\\
\begin{theorem}
Let $\overline{\sigma}$ be a closed one form on $\overline{Q}$. Then, the following conditions are equivalent:

\begin{itemize}
\item[(i)] $\overline{\Gamma}_{H,M}^{\overline{\sigma}}$ and $\overline{\Gamma}_{H,M}$ are $\overline{\sigma}$-related.

\item[(ii)] $d \left( \overline{H} \circ \overline{\sigma} \right) = \overline{\sigma}^{*}\overline{\theta}$
\end{itemize} 
\end{theorem}

\subsection{Horizontal case}

Now, we will present a Hamilton-Jacobi theorem for the horizontal case. First, let us fix some notations: Let $\tau_{Q}$ be the tangent bundle of $Q$,
\begin{itemize}
\item $\tau_{Q,\mu} : \left( TQ \right)_{\mu} \rightarrow Q_{\mu}= Q / G_{\mu}$ is the projection of $\tau_{Q}$.
\item $\tau_{Q,\mu} \circ j_{L,\mu} = \tau_{\mu} : \mathcal{P}_{\mu} \rightarrow Q_{\mu}$.
\end{itemize}
For any section $X_{\mu}$ of $\tau_{\mu}$ such that $X_{\mu} \left( Q_{\mu} \right) \subseteq \mathcal{N}_{\mu}$ induces a vector field $\left( \Gamma_{L,N} \right)_{\mu}^{X}$ on $Q_{\mu}$ by projecting $\left( \Gamma_{L,N} \right)_{\mu}$ following the next diagram,
\begin{equation} \label{hattau3}
  \xymatrix{ \mathcal{N}_{\mu} \ar"1,4"^{\left( \Gamma_{L,N} \right)_{\mu}}&   & &T \mathcal{N}_{\mu}  \ar[dd]^{T\left(\tau_\mu\vert_{\mathcal{N}_{\mu}}\right)}\\
  &  & &\\
 Q_{\mu} \ar"3,4"^{\left( \Gamma_{L,N} \right)_{\mu}^{X}} \ar"1,1"^{X_{\mu}}&  & & TQ_{\mu} }
\end{equation}
We can now consider an ideal $\mathcal{I} \left( F^{o}_{L,\mu} \right)$ in the algebra of forms on $\mathcal{P}_{\mu}$ given by forms of the following type (see the generating forms given in \ref{aux34d})
$$ \beta_{i}\wedge  j_{L,\mu}^{*} S_{\mu}^{*} \left(  \overline{d\psi^{a}} \right), \ \beta_{i} \in \Lambda^{k}\left( \mathcal{P}_{\mu} \right).$$

\begin{theorem}\label{Hor234sd}
Let $X_{\mu}$ be a section of $\tau_{\mu}$ such that $X_{\mu} \left( Q_{\mu} \right) \subseteq \mathcal{N}_{\mu}$ and $\left(X_{\mu} \circ \tau_{\mu}\right)^{*} \omega_{L,\mu} \in \mathcal{I} \left(F^{o}_{L,\mu}\right)$. Then, the following conditions are equivalent:

\begin{itemize}
\item[(i)] $\left(\Gamma_{L,N} \right)_{\mu}^{X}$ and $\left(\Gamma_{L,N} \right)_{\mu}$ are $X_{\mu}-$related.

\item[(ii)] $d \left( \left(E_{L}\right)_{\mu} \circ \left( X_{\mu} \circ \tau_{\mu}\right) \right) \in F^{o}_{L,\mu}$
\end{itemize} 

\begin{proof}
Notice that condition $\left(X_{\mu} \circ \tau_{\mu}\right)^{*} \omega_{L,\mu} \in \mathcal{I} \left(F^{o}_{L,\mu}\right)$ implies that
\begin{equation}\label{135}
\iota_{\left(\Gamma_{L,N} \right)_{\mu}} \left(X_{\mu} \circ \tau_{\mu}\right)^{*} \omega_{L,\mu} \in F^{o}_{L,\mu}
\end{equation}
along $X_{\mu}$.\\
Now, let us prove that
\begin{equation}\label{137}
\left( X_{\mu} \circ \tau_{\mu}\right)^{*}F^{o}_{L,\mu} =  F^{o}_{L,\mu},
\end{equation}
along $X_{\mu}$.\\
Again (see section \ref{Generalcase3}), there exists a vector field $X$ on $Q$ such that,
$$ X_{\mu} \circ \rho_{Q,\mu} = \rho_{L,\mu} \circ X,$$
where $\rho_{Q,\mu}: Q \rightarrow Q_{\mu}$ is the canonical projection. Then, for each $a$, by using Eq. (\ref{132}), we have that

\begin{eqnarray*}
\rho_{L,\mu}^{*}\left[   \left(X_{\mu} \circ \tau_{\mu}\right)^{*} j_{L,\mu}^{*} S^{*}_{\mu} \left(  \overline{d\psi^{a}} \right) \right] &=& \left(X_{\mu} \circ \tau_{\mu} \circ \rho_{L,\mu}\right)^{*} j_{L,\mu}^{*} S^{*}_{\mu} \left(   \overline{d\psi^{a}} \right)\\
&=&\left[ X \circ \tau_{Q} \circ i_{L,\mu}\right]^{*} \rho_{L,\mu}^{*} \left(j_{L,\mu}^{*} S^{*}_{\mu} \left( \overline{d\psi^{a}}\right) \right)\\
&=& \left[ X \circ \tau_{Q} \circ i_{L,\mu} \right]^{*} \left(\left(j_{L,\mu} \circ \rho_{L,\mu}\right)^{*}S^{*}_{\mu} \left(  \overline{d\psi^{a}}\right)\right)\\
&=& \left[ X \circ \tau_{Q} \circ i_{L,\mu} \right]^{*} \left(\left(\rho_{TQ,\mu} \circ i_{L,\mu}\right)^{*}S^{*}_{\mu} \left(  \overline{d\psi^{a}} \right)\right)\\
&=& \left[i_{L,\mu} \circ  X \circ \tau_{Q} \circ i_{L,\mu} \right]^{*} \left(\rho_{TQ,\mu}^{*}S^{*}_{\mu} \left(  \overline{d\psi^{a}} \right)\right)\\
&=& \left[i_{L,\mu} \circ  X \circ \tau_{Q} \circ i_{L,\mu} \right]^{*} S^{*} \left(\rho_{TQ,\mu}^{*} \left(  \overline{d\psi^{a}} \right)\right)\\
&=& i_{L,\mu}^{*}S^{*} \left(\rho_{TQ,\mu}^{*} \left(  \overline{d\psi^{a}} \right)\right) \\
&=& \rho_{L,\mu}^{*}\left[j_{L,\mu}^{*} S^{*}_{\mu} \left(  \overline{d\psi^{a}} \right)\right]
\end{eqnarray*}
Hence, for all $a$
$$\left(X_{\mu} \circ \tau_{\mu} \right)^{*}\left[ j_{L,\mu}^{*} S^{*}_{\mu} \left(  \overline{d\psi^{a}} \right)\right] = j_{L,\mu}^{*} S^{*}_{\mu} \left(  \overline{d\psi^{a}} \right).$$
Notice that we are using that,
$$   \left[i_{L,\mu} \circ  X \circ \tau_{Q} \circ i_{L,\mu} \right]^{*} S^{*} = i_{L,\mu}^{*}S^{*},$$
which can be proved using coordinates.\\
Therefore, Eq. (\ref{137}) and Eq. (\ref{136}) imply that $ \left(X_{\mu} \circ \tau_{\mu} \right)^{*}\left[ \iota_{\left(\Gamma_{L,N} \right)_{\mu}} \omega_{L,\mu} \right] \in F^{o}_{L,\mu}$ if, and only if, $d \left( \left(E_{L}\right)_{\mu} \circ \left( X_{\mu} \circ \tau_{\mu}\right) \right) \in F^{o}_{L,\mu}$. So, by taking into account Eq. (\ref{135}) we can easily prove that $\left(i\right)$ implies $\left( ii \right)$.\\
Conversely, notice that
$$\left(\Gamma_{L,N} \right)_{\mu} \circ X_{\mu} = T \left( X_{\mu} \circ \tau_{\mu} \right) \circ \left(\Gamma_{L,N} \right)_{\mu} \circ X_{\mu} + V_{\mu},$$
where $V_{\mu} \in Ker \left( T \tau_{\mu} \right)$ for each point. So, $V_{\mu} \in T \mathcal{N}_{\mu}$.\\
On the other hand, by using conditions $\left(X_{\mu} \circ \tau_{\mu}\right)^{*} \omega_{L,\mu} \in \mathcal{I} \left(F^{o}_{L,\mu}\right)$ and \linebreak $\left(X_{\mu} \circ \tau_{\mu} \right)^{*}\left[ \iota_{\left(\Gamma_{L,N} \right)_{\mu}} \omega_{L,\mu} \right] \in F^{o}_{L,\mu}$ we have that

\begin{eqnarray*}
\left( X_{\mu} \circ \tau_{\mu} \right)^{*} \left[ \iota_{V_{\mu}} \omega_{L,\mu}\right] & = & \left( X_{\mu} \circ \tau_{\mu} \right)^{*} \left[ \iota_{\left(\Gamma_{L,N} \right)_{\mu} \circ X_{\mu}  -   T \left( X_{\mu} \circ \tau_{\mu} \right) \circ \left(\Gamma_{L,N} \right)_{\mu} \circ X_{\mu} }\omega_{L,\mu}\right] \\
&=&   \left( X_{\mu} \circ \tau_{\mu} \right)^{*} \left[ \iota_{\left(\Gamma_{L,N} \right)_{\mu} \circ X_{\mu}} \omega_{L,\mu} \right]  \\
&-&   \left( X_{\mu} \circ \tau_{\mu} \right)^{*} \left[ \iota_{T \left( X_{\mu} \circ \tau_{\mu} \right) \circ \left(\Gamma_{L,N} \right)_{\mu} \circ X_{\mu} } \omega_{L,\mu} \right]\\
&=& \left( X_{\mu} \circ \tau_{\mu} \right)^{*} \left[ \iota_{\left(\Gamma_{L,N} \right)_{\mu} \circ X_{\mu}} \omega_{L,\mu} \right]  \\
&-&   \iota_{\left(\Gamma_{L,N} \right)_{\mu} \circ X_{\mu}} \left[ \left( X_{\mu} \circ \tau_{\mu} \right)^{*} \omega_{L,\mu}  \right] \in F^{o}_{L,\mu}.
\end{eqnarray*}
Then,
\begin{equation}\label{142}
\omega_{L,\mu} \left( V_{\mu} , T \left( X_{\mu} \circ \tau_{\mu} \right) \left( Z \right) \right) = 0, \ \forall Z \in F_{L,\mu}.
\end{equation}

Let us now consider $V \in F_{L,\mu} \cap Ker \left( T \tau_{\mu} \right)$. Hence, by using that $V_{\mu} \in Ker \left( T \tau_{\mu} \right)$ we have that
\begin{equation}\label{138}
\omega_{L,\mu} \left( V_{\mu} , V \right) = 0.
\end{equation}
Let us prove this identity carefully. Consider $U,W \in T \left( TQ \right)$ such that 
\begin{equation}\label{139}
T\rho_{Q,\mu} \left(T \tau_{Q} \left( U \right)\right)= T\rho_{Q,\mu} \left(T \tau_{Q} \left( W \right)\right) = 0.
\end{equation}
Then,
\begin{equation}\label{141}
\omega_{L} \left( U,W \right) =0.
\end{equation}
In fact, let $\left( q^{i} , x^{j} \right)$ be a local system of coordinates on $Q$ adapted to $\rho_{Q,\mu}$, i.e.,
$$ \rho_{Q,\mu} \left(  q^{i} , x^{j} \right) = q^{i}.$$
Then, the family of vectors satisfying Eq. (\ref{139}) is generated by the partial derivatives $\dfrac{\partial}{\partial x^{i}}$, $\dfrac{\partial}{\partial \dot{x}^{i}}$ and $\dfrac{\partial}{\partial \dot{q}^{j}}$. Next, as a result of the $G-$invariance, for all $i$
$$ dL \left( \dfrac{\partial}{\partial x^{i}} \right) = 0,$$
i.e., 
\begin{equation}\label{140}
\dfrac{\partial L }{\partial x^{i}} = 0, \ \forall i.
\end{equation}
Then we only have to use the local expression of $\omega_{L}$ to prove Eq. (\ref{141}) and, hence, Eq. (\ref{138}). Notice that Eq. (\ref{138}) is slightly more general than Lemma \ref{omega-T}.
Let be $Z \in F_{L,\mu}$ at a point along $X_{\mu}$. Then,
$$ Z = T \left( X_{\mu} \circ \tau_{\mu} \right) \left( Z \right) + V,$$
with $V  \in F_{L,\mu} \cap Ker \left( T \tau_{\mu} \right)$. Here we are using that $T \left( X_{\mu} \circ \tau_{\mu} \right)$ preserves $F_{L,\mu}$ via Eq. (\ref{137}). Then, by taking into account Eq. (\ref{142}) and Eq. (\ref{138}), we have that
\begin{equation}\label{143}
\omega_{L,\mu} \left( V_{\mu} , Z \right) = 0.
\end{equation}
i.e., $V_{\mu} \in T \mathcal{N}_{\mu} \cap F_{L,\mu}^{\bot} = \{0\}$ and, therefore,
$$\left(\Gamma_{L,N} \right)_{\mu} \circ X_{\mu} = T \left( X_{\mu} \circ \tau_{\mu} \right) \circ \left(\Gamma_{L,N} \right)_{\mu} \circ X_{\mu},$$
Which proves $(i)$.
\end{proof}
\end{theorem}

Let us highlight a particular case. Suppose that
\begin{equation}\label{144}
S \left( TN \right) \subseteq TN.
\end{equation}
Then, $F_{L} = TN$. Therefore the following identity
$$\mathcal{V}_{N} \cap \mathcal{H} = \mathcal{V}_{N},$$
is trivially satisfied, i.e., this case is an horizontal symmetries case. Notice that Eq. (\ref{144}) can be interpreted as a compatibility condition between the constraints and the vertical endomorphism. It is easy to check that $F_{L} = TN$ implies Eq. (\ref{144}).\\

\begin{definition}
A section $X_{\mu}$ of $\tau_{\mu}$ satisfying conditions of the Theorem \ref{Hor234sd} will be called $\mu-$solution for the reduced constrained Hamilton-Jacobi problem given by the constraint manifold $N \subseteq TQ$, the regular Lagrangian $L$ and the action of the group $G$ on $Q$.
\end{definition}

It is remarkable that, as we have noticed in the proof of the theorem, there is a one-to-one correspondence between $G_{\mu}-$invariant vector fields $X$ on $Q$ such that $X \left( Q \right) \subseteq J^{-1}_{L} \left( \mu \right)$ and sections $X_{\mu}$ of $\tau_{\mu}$. This correspondence will be denoted by
$$X \mapsto X_{\mu}.$$
The proof of this correspondence is analagous to the correspondence $X \mapsto \overline{X}$ (see (\ref{reconstrucion23}). Thus, it works as a reconstruction process from $\mu-$solutions for the reduced constrained Hamilton-Jacobi problem to solutions for the constrained Hamilton-Jacobi problem.
\begin{proposition}
Let $X$ be a $G_{\mu}-$invariant vector field on $Q$ such that $X \left( Q \right) \subseteq J^{-1}_{L} \left( \mu \right)$. Then, $X$ is a solution for the constrained Hamilton-Jacobi problem if, and only if, $X_{\mu}$ is a $\mu-$solution for the reduced constrained Hamilton-Jacobi problem.
\end{proposition}

Let us now study analogous consequences to Corollary \ref{funnycorollary123} and Theorem \ref{AppChowtheore,}.
\begin{corollary}\label{funnycorollary123Horizontalcase}
Let $X_{\mu}$ be a section of $\tau_{\mu}$ in the conditions of Theorem \ref{Hor234sd} such that $T X_{\mu} \left( TQ_{\mu} \right) \subseteq F_{L,\mu}$. Then, the following conditions are equivalent:

\begin{itemize}
\item[(i)] $\left(\Gamma_{L,N} \right)_{\mu}^{X}$ and $\left(\Gamma_{L,N} \right)_{\mu}$ are $X_{\mu}-$related.

\item[(ii)] $d \left( \left(E_{L}\right)_{\mu} \circ  X_{\mu} \right) =0.$
\end{itemize} 
\begin{proof}
Notice that, because of the existence of the section $X_{\mu}$, $\tau_{\mu}$ is a submersion. Then, the $1-$form $d \left( \left(E_{L}\right)_{\mu} \circ  X_{\mu} \right)$ is zero if, and only if,
$$d \left( \left(E_{L}\right)_{\mu} \circ \left( X_{\mu} \circ \tau_{\mu} \right)\right) =\tau_{\mu}^{*} \left[ d \left( \left(E_{L}\right)_{\mu} \circ  X_{\mu} \right)  \right] = 0.$$
Finally, as an immediate consequence of the condition $T X_{\mu} \left( TQ_{\mu}\right) \subseteq F_{L,\mu}$, it is just an exercise to check that the only way in which the $1-$form \linebreak $d \left( \left(E_{L}\right)_{\mu} \circ \left( X_{\mu} \circ \tau_{\mu} \right)\right)$ is in the annihilator $F_{L,\mu}^{o}$ is that it is cancelled. 
\end{proof}
\end{corollary}

We will study now an application of the Chow-Rashevskii's Theorem (see \ref{LordChowTheorem}) to the reduction of the Horizontal case.\\
Notice, once again, that $F_{L,\mu}$ is not a distribution defined on the whole space $\mathcal{P}_{\mu}$ (it is only defined on $\mathcal{N}_{\mu}$. Thus, we will solve it in an analogous way.\\
Thus, we may define the distribution $G_{L,\mu}$ on $\mathcal{P}_{\mu}$ by resticting and quotienting $G_{L}$ (see Section \ref{firstresut23}). Obviously, $G_{L,\mu}$ resticts to $F_{L,\mu}$ over $\mathcal{N}_{\mu}$.\\ 
Let us consider the family of (local) vector fields $\mathcal{S}_{\mathcal{N}_{\mu}}^{\sharp}$ tangents to $G_{L,\mu}$ whose domains have non-empty intersection with $\mathcal{N}_{\mu}$. Then, $F_{L,\mu}$ is the distribution generated by $\mathcal{S}_{\mathcal{N}_{\mu}}^{\sharp}$ along $\mathcal{N}_{\mu}$.\\
Thus, $F_{L,\mu}$ is called completely nonholonomic if the reiterated Lie bracket of vector fields in $\mathcal{S}_{\mathcal{N}_{\mu}}^{\sharp}$ generate $T_{\mathcal{N}_{\mu}} \mathcal{P}_{\mu}$. In this case we will say that \linebreak$\left[F_{L,\mu},F_{L,\mu}\right], \left[ F_{L,\mu}, \left[ F_{L,\mu},F_{L,\mu} \right] \right], \hdots$ spans the vector bundle $T_{\mathcal{N}_{\mu}}\mathcal{P}_{\mu}$.
Obviously, $G_{L,\mu}$ is completely nonholonomic if, and only if, $F_{L,\mu}$ is completely nonholonomic.\\
we are now ready to give a reduced version of Theorem \ref{AppChowtheore,}.
\begin{theorem}\label{AppChowtheore,Horizcase34}
Let $X_{\mu}$ be a section of $\tau_{\mu}$ in the conditions of Theorem \ref{Hor234sd} such that $F_{L,\mu}$ is completely nonholonomic. Then, the following conditions are equivalent:
\begin{itemize}
\item[(i)] $\left(\Gamma_{L,N} \right)_{\mu}^{X}$ and $\left(\Gamma_{L,N} \right)_{\mu}$ are $X_{\mu}-$related.

\item[(ii)] $d \left( \left(E_{L}\right)_{\mu} \circ  X_{\mu} \right) =0.$
\end{itemize}
\end{theorem}
\begin{proof}
Due to the fact that $G_{L,\mu}$ is completely nonholonomic we can use Proposition \ref{Importantconsequuence2134} to have that $d  \left( \left(E_{L}\right)_{\mu} \circ  X_{\mu} \right) \in G_{L,\mu}$ if, and only if, \linebreak $d  \left( \left(E_{L}\right)_{\mu} \circ  X_{\mu} \right)=0$.
\end{proof}

Finally, following a similar development to the General case, we will study the relation between the distributions $F^{'}_{L}$ and $F_{L,\mu}$.\\

\begin{proposition}
$F^{'}_{L}$ is completely nonholonomic if, and only if, $F_{L,\mu}$ is completely nonintegrable.
\end{proposition}
\begin{proof}
To prove this result we only have to change the group $G$ by the the isotropy group $G_{\mu}$ in the proof of Proposition \ref{Relation23}.
\end{proof}

Note that condition of being completely nonintegrable over $F^{'}_{L}$ is not, generally, related with the same condition over $F_{L}$. This is a consequence of that, in general, the momentum map $J_{L}$ does not have relation with the vertical endomorphism $S$ and the constraint manifold $N$.

Now, let us study the Hamiltonian case via the Legrendre transformation $\mathbb{F}L$. Recall that $\mathbb{F}L \left( N \right) = M$, $H = E_{L} \circ \mathbb{F}L^{-1}$, $ \omega_{Q}={\mathbb{F}L^{-1}}^{*} \omega_{L}$ and $\Gamma_{H,M} = {\mathbb{F}L}_{*} \Gamma_{L,N} $.\\
Next, consider $J_{H}: T^{*}Q \rightarrow \mathfrak{g}^{*}$ the classical momentum map of the contangent bundle. Then, $J_{H} = J_{L} \circ \mathbb{F}L^{-1}$.\\
\noindent
Again, we have a reduced sympletic manifold
$$ \left( \mathcal{Q}_{\mu} = J_{H}^{-1} \left( \mu \right) / G_{\mu}, \omega_{Q,\mu} \right),$$
where $G_{\mu}$ is the isotropy group of $\mu$ for the co-adjoint action and $\omega_{Q,\mu}$ is the unique symplectic $2-$form on $\mathcal{Q}_{\mu}$ such that
$$ \rho_{H,\mu}^{*} \omega_{Q,\mu} = i_{H,\mu}^{*} \omega_{Q},$$
with $\rho_{H,\mu} : J_{H}^{-1} \left( \mu \right) \rightarrow \mathcal{Q}_{\mu}$ is the canonical projection and $i_{H,\mu}: J_{H}^{-1}\left( \mu \right) \hookrightarrow T^{*}Q$ is the inclusion map.\\
Notice that $\mathbb{F}L$ can be restricted to a map from $J_{L}^{-1} \left( \mu \right)$ to $J_{H}^{-1} \left( \mu \right)$. In fact, for each vector $v_{q} \in J_{L}^{-1} \left( \mu \right)$ we have
$$J_{H} \left( \mathbb{F}L \left( v_{q} \right) \right)= J_{L} \left( v_{q} \right) = \mu.$$
So, quotienting the restiction, we get a map $\mathbb{F}L_{\mu}$ from $\mathcal{Q}_{\mu}$ to $\mathcal{P}_{\mu}$ characterized by the following identity
\begin{equation}\label{147}
\mathbb{F}L_{\mu} \circ \rho_{L,\mu} = \rho_{H,\mu} \circ \mathbb{F}L_{|J_{L}^{-1} \left( \mu \right)}.
\end{equation}
Denote the intersection of $M$ and $J_{H}^{-1} \left( \mu \right)$ by $M'$. $M^{'}$ is obviously $G_{\mu}-$invariant and, hence, it projects onto a submanifold $\mathcal{M}_{\mu}$ of $\mathcal{Q}_{\mu}$.\\
Notice that $M' = \mathbb{F}L \left( N' \right)$ and, hence,
$$ \mathcal{M}_{\mu} = \mathbb{F}L_{\mu} \left( \mathcal{N}_{\mu} \right).$$


Let us again fix some notations: Let $\pi_{Q}$ be the cotangent bundle of $Q$,
\begin{itemize}
\item $\pi_{Q,\mu} : \left( T^{*}Q \right)_{\mu} \rightarrow Q_{\mu}= Q / G_{\mu}$ is the projection of $\pi_{Q}$.
\item $j_{H,\mu} \circ \pi_{Q,\mu} = \pi_{\mu} : \mathcal{Q}_{\mu} \rightarrow Q_{\mu}$, where $j_{H,\mu} : \mathcal{Q}_{\mu} \hookrightarrow \left( T^{*}Q \right)_{\mu}$ is the inclusion map.
\end{itemize}

For any section $\sigma_{\mu}$ of $\pi_{\mu}$ such that $\sigma_{\mu} \left( Q_{\mu} \right) \subseteq \mathcal{M}_{\mu}$ induces a vector field $\left( \Gamma_{H,M} \right)_{\mu}^{\sigma}$ on $Q_{\mu}$ by projecting $\left( \Gamma_{H,M} \right)_{\mu}$ following the next diagram,
\begin{equation} \label{hattau4}
  \xymatrix{ \mathcal{M}_{\mu} \ar"1,4"^{\left( \Gamma_{H,M} \right)_{\mu}}&   & &T \mathcal{M}_{\mu}  \ar[dd]^{T\left(\pi_\mu\vert_{\mathcal{M}_{\mu}}\right)}\\
  &  & &\\
 Q_{\mu} \ar"3,4"^{\left( \Gamma_{H,M} \right)_{\mu}^{\sigma}} \ar"1,1"^{\sigma_{\mu}}&  & & TQ_{\mu} }
\end{equation}

\begin{theorem}\label{Th1}
Let $\sigma_{\mu}$ be a section of $\pi_{\mu}$ such that $\sigma_{\mu} \left( Q_{\mu} \right) \subseteq \mathcal{M}_{\mu}$ and $\left(\sigma_{\mu} \circ \tau_{\mu}\right)^{*} \omega_{Q,\mu} \in \mathcal{I} \left(F^{o}_{L,\mu}\right)$. Then, the following conditions are equivalent:

\begin{itemize}
\item[(i)] $\left(\Gamma_{H,M} \right)_{\mu}^{\sigma}$ and $\left(\Gamma_{H,M} \right)_{\mu}$ are $\sigma_{\mu}-$related.

\item[(ii)] $d \left( H_{\mu} \circ \left( \sigma_{\mu} \circ \tau_{\mu}\right) \right) \in F^{o}_{L,\mu}$
\end{itemize} 
\end{theorem}

We will now give a Hamiltonian version of Corollary \ref{funnycorollary123Horizontalcase} and Theorem \ref{AppChowtheore,Horizcase34}.

\begin{theorem}\label{mmmmmm}
Let $\sigma_{\mu}$ be a $1-$form on $Q$ such that $\sigma_{\mu} \left( \overline{Q} \right) \subset \overline{M}$ and $\left(\sigma_{\mu} \circ \overline{\tau}_{Q}\right)^{*} \overline{\omega}_{Q} \in \mathcal{I} \left(\overline{F}^{o}_{L}\right)$. Then, the following conditions are equivalent:

\begin{itemize}
\item[(i)] $\overline{\Gamma}_{H,M}^{\sigma_{\mu}}$ and $\overline{\Gamma}_{H,M} $ are $\sigma_{\mu}$-related.
\item[(ii)] $d \left( \overline{H} \circ \left( \sigma_{\mu} \circ \overline{\tau}_{Q}\right) \right) \in \overline{F}^{o}_{L}$
\end{itemize} 
\end{theorem}

\begin{definition}
A section $\sigma_{\mu}$ of $\pi_{\mu}$ satisfying conditions of the Theorem \ref{mmmmmm} will be called $\mu-$solution for the reduced constrained Hamilton-Jacobi problem given by the constraint manifold $M \subseteq T^{*}Q$, the Hamiltonian $H$ and the action of the group $G$ on $Q$.
\end{definition}

One again, the space of $G_{\mu}-$invariant $1-$forms $\sigma$ on $Q$ such that $\sigma \left( Q \right) \subseteq J_{H}^{-1} \left( \mu \right)$ are in one-to-one correspondence with sections $\sigma_{\mu}$ of $\pi_{\mu}$
\begin{proposition}
Let $\sigma$ be a $G_{\mu}-$invariant $1-$form on $Q$ such that $\sigma \left( Q \right) \subseteq J_{H}^{-1} \left( \mu \right)$. Then, $\sigma$ is a solution for the constrained Hamilton-Jacobi problem if, and only if, $\sigma_{\mu}$ is a $\mu-$solution for the reduced constrained Hamilton-Jacobi problem.
\end{proposition}

\begin{corollary}
Let $\sigma_{\mu}$ be a section of $\pi_{\mu}$ in the condition of Theorem \ref{mmmmmm} satisfying $T\sigma_{\mu} \left( TQ_{\mu}\right) \subseteq F_{\mu}$.
Then, the following conditions are equivalent:
\begin{itemize}
\item[(i)] $\left(\Gamma_{H,N} \right)_{\mu}^{\sigma}$ and $\left(\Gamma_{H,N} \right)_{\mu}$ are $\sigma_{\mu}$-related.
\item[(ii)] $d \left( H_{\mu} \circ \sigma_{\mu} \right) = 0$
\end{itemize}
\end{corollary}
\noindent
Here, $F_{\mu}$ is the transformed distribution of $F_{L}$ via $\overline{\mathbb{F}L}$.\\
Finally, we can give a Hamiltonian version of the Theorem \ref{mmmmmm}.

\begin{theorem}\label{AppChowtheoreprojected22dualversionnHor}
Let $\sigma_{\mu}$ be a section of $\pi_{\mu}$ in the condition of Theorem \ref{mmmmmm} such that $F_{L,\mu}$ is completely nonholonomic. Then, the following conditions are equivalent:
\begin{itemize}
\item[(i)] $\left(\Gamma_{H,N} \right)_{\mu}^{\sigma}$ and $\left(\Gamma_{H,N} \right)_{\mu}$ are $\sigma_{\mu}$-related.
\item[(ii)] $d \left( H_{\mu} \circ \sigma_{\mu} \right) = 0$
\end{itemize} 
\end{theorem}

As a particular case, let us consider again the case in which $N_{\ell}$ is a vector subbundle of $TQ$ and the constraint functions $\psi^{a}$ are linear, i.e., $\psi^{a}$ induce $1-$forms $\overline{\psi}^{a}$ on $Q$ which generate $N_{\ell}^{o}$ (see \ref{linearconstraints23323}).\\

Observe that the following diagram is comutative (see (\ref{linearconstraints23323})),

\begin{equation} \label{hattau5}
  \xymatrix{ N_{\ell} \ar"1,4"^{F_{L}^{o}} \ar"3,1"^{\tau_{Q}}&   & &T^{*}\left( TQ \right)  \\
  &  & &\\
 Q \ar"3,4"^{N^{o}_{\ell}}&  & & T^{*}Q\ar"1,4"^{T^{*}\tau_{Q}} }
\end{equation}

i.e.,
$$\tau_{Q}^{*}\left( N_{\ell}^{o}\right) = F_{L}^{o}.$$

Let us prove, in this case, the Hamiltonian version of the corollary.

\begin{corollary}
Let $\sigma$ be a $G_{\mu}-$invariant $1-$form on $Q$ such that $\sigma \left( Q \right) \subseteq N_{\ell}$ and $d \sigma \in \mathcal{I} \left( N_{\ell}^{o} \right)$. Let $\sigma_{\mu}$ be the projection of $\sigma$ as a section of $\pi_{\mu}$. Then, the following conditions are equivalent:

\begin{itemize}
\item[(i)] $\left(\Gamma_{H,N} \right)_{\mu}^{\sigma}$ and $\left(\Gamma_{H,N} \right)_{\mu}$ are $\sigma_{\mu}-$related.

\item[(ii)] $d \left( H_{\mu} \circ \left( \sigma_{\mu} \circ \tau_{\mu}\right) \right) \in F^{o}_{L,\mu}$
\end{itemize} 
\begin{proof}
Let $\sigma$ be a $G_{\mu}-$invariant $1-$form on $Q$ such that $\sigma \left( Q \right) \subseteq N_{\ell}$ and $d \sigma \in \mathcal{I} \left( N_{\ell}^{o} \right)$. Condition $d \sigma \in \mathcal{I} \left( N_{\ell}^{o} \right)$ implies that
$$\tau_{Q}^{*}\left[d \sigma \right] = d \left( \tau_{Q}^{*}\sigma \right) = \left( \sigma \circ \tau_{Q} \right)^{*}\omega_{Q} \in \mathcal{I} \left( F_{L}^{o} \right).$$
This is a consequence of the commutativity of diagram (\ref{hattau5}).\\
Consider the projection $\sigma_{\mu}$ of $\sigma$ as a section of $\pi_{\mu}$. Then, with some slight abuse of notation we have

\begin{eqnarray*}
\rho_{L,\mu}^{*} \left[ \left( \sigma_{\mu} \circ \tau_{\mu} \right)^{*}\omega_{Q,\mu} \right] &=& \left(  \sigma_{\mu} \circ \tau_{\mu} \circ \rho_{L,\mu} \right)^{*}\omega_{Q,\mu}\\
&=& \left(  \sigma_{\mu} \circ \rho_{Q,\mu} \circ \tau_{Q|J_{L}^{-1}\left(\mu\right)} \right)^{*}\omega_{Q,\mu}\\
&=& \left( \rho_{L,\mu} \circ \sigma \circ \tau_{Q|J_{L}^{-1}\left(\mu\right)} \right)^{*}\omega_{Q,\mu}\\
&=& \left( \sigma \circ \tau_{Q|J_{L}^{-1}\left(\mu\right)} \right)^{*}\left[ \rho_{L,\mu}^{*}\omega_{Q,\mu}\right]\\
&=& \left( \sigma \circ \tau_{Q|J_{L}^{-1}\left(\mu\right)} \right)^{*}\left[ i_{L,\mu}^{*}\omega_{Q}\right]\\
&=& i_{L,\mu}^{*}\left[\left( \sigma \circ \tau_{Q} \right)^{*}\omega_{Q}\right] \in i_{L,\mu}^{*}\mathcal{I} \left( F_{L}^{o} \right),\\
\end{eqnarray*}
Notice that $F^{o}_{L,\mu}$ is characterized by te condition
$$ \rho_{L,\mu}^{*} F^{o}_{L,\mu} = i_{L, \mu}^{*}F^{o}_{L},$$
and, therefore, $\left( \sigma_{\mu} \circ \tau_{\mu} \right)^{*}\omega_{Q,\mu}  \in \mathcal{I}\left( F^{o}_{L,\mu}\right)$, i.e., we are in the conditions of Theorem \ref{Th1}.\\
\end{proof}
\end{corollary}

Let us study the condition $\left( ii\right)$, i.e.,
$$d \left( H_{\mu} \circ \left( \sigma_{\mu} \circ \tau_{\mu}\right) \right) \in F^{o}_{L,\mu},$$
which is equivalent to 
$$d \left( H_{\mu} \circ \left( \sigma_{\mu} \circ \tau_{\mu}\right) \circ \rho_{L,\mu} \right) \in i_{L, \mu}^{*} F^{o}_{L}.$$
Observe that

\begin{eqnarray*}
d \left( H_{\mu} \circ \left( \sigma_{\mu} \circ \tau_{\mu}\right) \circ \rho_{L,\mu} \right) &=& 
d \left( H_{\mu} \circ \sigma_{\mu} \circ \tau_{Q,\mu} \circ i_{L,\mu} \right)\\
&=& d \left( H_{\mu} \circ \sigma_{\mu} \circ \rho_{Q,\mu} \circ \tau_{Q |J_{L}^{-1}\left(\mu \right)} \right)\\
&=& d \left( H_{\mu} \circ \rho_{\mu} \circ \sigma \circ \tau_{Q |J_{L}^{-1}\left(\mu \right)} \right)\\
&=& d \left( H \circ \sigma \circ \tau_{Q |J_{L}^{-1}\left(\mu \right)} \right)\\
&=& \left( \tau_{Q} \circ i_{L,\mu}\right)^{*}\left[d \left( H \circ \sigma \right) \right]\\
&=& i_{L,\mu}^{*}\left[\tau_{Q}^{*}d \left( H \circ \sigma \right) \right]
\end{eqnarray*}

Therefore, $d \left( H \circ \sigma \right) \in N_{\ell}^{o}$ implies condition $\left( ii \right)$ of the above corollary.\\

Finally, we have the following corollary for the Lagrangian counterpart.
\begin{corollary}
Let $X$ be a $G_{\mu}-$invariant vector field on $Q$ such that $X \left( Q \right) \subseteq N$ and $d \left( \mathbb{F}L \circ X \right) \in \mathcal{I} \left( N_{\ell}^{o} \right)$. Let $X_{\mu}$ be the projection of $X$ as a section of $\tau_{\mu}$. Then, the following conditions are equivalent:

\begin{itemize}
\item[(i)]  $\left(\Gamma_{L,N} \right)_{\mu}^{X}$ ad $\left(\Gamma_{L,N} \right)_{\mu}$ are $X_{\mu}-$related.

\item[(ii)] $d \left( \left(E_{L}\right)_{\mu} \circ \left( X_{\mu} \circ \tau_{\mu}\right) \right) \in F^{o}_{L,\mu}$
\end{itemize} 
\end{corollary}

\section{Examples}

This section is reserved for two concrete examples of Hamilton-Jacobi
formalisms of nonholonomic systems admitting some symmetries.

\subsection{Free Particle with nonholonomic constraint}\label{firstexample324}

The first example is a free particle under a nonholomic constraint \cite{BatesSnia}. We cite \cite{de2014hamilton} for the Hamilton-Jacobi theory of this system in the realm of Poisson framework. In this paper, we discuss the Hamilton-Jacobi formalism of the system after a symmetry reduction. Consider a three dimensional space $Q=%
\mathbb{R}
^{3}$ with coordinates $\left( x,y,z\right)$, and its tangent bundle $TQ=%
\mathbb{R}
^{3}\times 
\mathbb{R}
^{3}$ with induced coordinates $\left( x,y,z,\dot{x},\dot{y},\dot{z}\right)$.
Lagrangian function is the kinetic energy of a single particle with mass $m$ that is
\[
L=\frac{1}{2}m\left( \dot{x}^{2}+\dot{y}^{2}+\dot{z}^{2}\right) 
\]
whereas the nonholonomic submanifold is defined to be
\[
N_{\ell }=\{\left( x,y,z,\dot{x},\dot{y},\dot{z}\right) \in TQ:\psi =\dot{z}%
-y\dot{x}=0\}. 
\]%
Using the linear constraint function $\psi $ we introduce a
differential one-form $\bar{\psi}=dz-ydx$. In accordance with this, $N_{\ell }$ can be defined as the space linearly generated by two linear independent vector fields $\xi_{1}$ and $\xi_{2}$,
\[
N_{\ell }=\left\langle \{\xi_1=\dfrac{\partial}{\partial x}+y\dfrac{\partial}{\partial z},\xi_2=\dfrac{\partial}{\partial y}\} \right \rangle.
\]%
The Jacobi-Lie bracket $[\xi_1,\xi_2] = - \xi_3 = -\dfrac{\partial}{\partial z}$ is a vector field such that $\xi_{1}$, $\xi_{2}$ and $\xi_{3}$ are linearly independent. So, these three vector fields span the tangent bundle $TQ$ so that we can say that $N_{\ell }$ is a completely nonholonomic constraint submanifold satisfying Definition (\ref{complete-nonholonomic}). The Hamilton-Jacobi equation for this picture can be obtained by the Hamilton-Jacobi theorem (\ref{nhhj21corollary4}). For this case, we search for a solution $X$ of the relation 
\begin{equation} \label{HJ-ex-1-unreduced}
d(E_L\circ X)=0,
\end{equation}
provided $X(Q)\subset N_\ell$. 
Instead of solving $X$ directly from the present setting, we apply  reduction procedure to the system, and compute $X$ as the $G$-invariant extension of the solution $\overline{X}$ of the reduced Hamilton-Jacobi equation (\ref{SecondThofred}). 

Translation $xz-$plane in $Q$ under the group action of $G=%
\mathbb{R}^{2}$ is given by%
\[
G\times Q:\left( \left( r,s\right) ,\left( x,y,z\right) \right) \rightarrow
\left( x+r,y,z+s\right) . 
\]%
It is an easy exercise to check that the tangent lift of this action preserves $L$ and $N_{\ell}$. This results with a quotient manifold $\overline{Q}\simeq 
\mathbb{R} $ with coordinate $\left( y\right) ,$ and the reduced tangent bundle is $%
\overline{TQ}\simeq 
\mathbb{R}
^{4}$ which is cover by coordinates $\left( y,\dot{x},\dot{y},\dot{z}\right)
.$ 
The distribution $F_{L}$ is computed as 
\begin{equation} \label{F-L-ex}
F_{L}= \left \langle\{\dfrac{\partial}{\partial x}+y\dfrac{\partial}{\partial z},\dfrac{\partial}{\partial y},\dfrac{\partial}{\partial \dot{x}%
},\dfrac{\partial}{\partial \dot{y}},\dfrac{\partial}{\partial \dot{z}}\} \right \rangle, 
\end{equation}
so that the subbundle in (\ref{H}) turns out to be%
\begin{equation} \label{H-ex}
\mathcal{H}=TN_{\ell }\cap F_{L}=\left \langle\{\dfrac{\partial}{\partial x}+y\dfrac{\partial}{\partial z},\dfrac{\partial}{\partial y}+\dot{x}\dfrac{\partial}{\partial \dot{z}},\dfrac{\partial}{\partial \dot{x}} + y\dfrac{\partial}{\partial \dot{z}},\dfrac{\partial}{\partial \dot{y}}\} \right \rangle.
\end{equation}
The vertical subbundle $\mathcal{V}_{N_{\ell}}$ is spanned by the pair of vector
fields $\dfrac{\partial}{\partial x},$ and $\dfrac{\partial}{\partial y}$. Intersection of $\mathcal{V}_{N_{\ell}}$ and $\mathcal{H}$ in (\ref{H-ex}) satisfies the following strict
inclusions 
\[
0\subsetneq \mathcal{V}_{N_{\ell}}\cap \mathcal{H}=\left \langle\{\dfrac{\partial}{\partial x}+y\dfrac{\partial}{\partial z}\} \right \rangle\subsetneq \mathcal{V}_{N_{\ell}} 
\]%
which says that we are in the general case introduced in the subsection (\ref{classification}).
Accordingly, we should apply one of the Hamilton-Jacobi theorems in subsection (\ref{Generalcase3}) to this present case.
Note that the reduction of the tangent bundle projection is 
\[
\overline{\tau }_{Q}:\overline{TQ}\simeq 
\mathbb{R}
^{4}\rightarrow \overline{Q}\simeq 
\mathbb{R}
:\left( y,\dot{x},\dot{y},\dot{z}\right) \rightarrow y 
\]%
so that a section $\overline{X}$ of this projection is a function 
\[
\overline{X}:\overline{Q}\simeq 
\mathbb{R}
\rightarrow \overline{TQ}\simeq 
\mathbb{R}
^{4}:y\rightarrow \left( y,\overline{A}\left( y\right) ,\overline{B}\left(
y\right) ,\overline{C}\left( y\right) \right) 
\]%
where $\overline{A}$, $\overline{B}$ and $%
\overline{C}$ are arbitrary functions depending on the free variable $y$. Reduction of the distribution ${F}_{L}$ in (\ref{F-L-ex}) under the group action results with
\[
\overline{F}_{L}=\left \langle\{\dfrac{\partial}{\partial y},\dfrac{\partial}{\partial \dot{x}},\dfrac{\partial}{\partial \dot{y}%
},\dfrac{\partial}{\partial \dot{z}}\} \right \rangle= T\overline{TQ}. 
\]
This reads that the condition $T\overline{X}\left( T\overline{Q}\right) \subset 
\overline{F}_{L}$, presented in the Corollary \ref{corollary-}, is trivially satisfied. We can use this Corollary in order to exhibit the
Hamilton-Jacobi formalism of the reduced system. The reduced constraint submanifold $\overline{%
N_{\ell }}$ is $3$ dimensional with coordinates $\left( y,\dot{x},\dot{y}%
\right) $ where $\dot{z}=y\dot{x}$. Therefore, if we insist that $\overline{X%
}\left( \overline{Q}\right) \subset \overline{N_{\ell }}$, we take that $%
\overline{C}=y\overline{A}$. The restiction of the reduced energy $\overline{E}_{L}$ to $\overline{N}_{\ell}$ is computed to be 
\[
\overline{E}_{L}:\overline{TQ}\rightarrow 
\mathbb{R}
:\left( y,\dot{x},\dot{y},\dot{z}\right) \rightarrow \frac{1}{2}m\left(
\left( 1+y^{2}\right) \dot{x}^{2}+\dot{y}^{2}\right). 
\]
We write the second condition $d(\overline{E}_{L}\circ \overline{X})=0$ in the Corollary \ref{corollary-} for the present case, and observe that
\[
\frac{1}{2}m\left( \left( 1+y^{2}\right) \overline{A}\left( y\right) ^{2}+%
\overline{B}\left( y\right) ^{2}\right) =c. 
\]%
Hence, a family of solutions of this equation can be given by $\overline{A}\left( y\right)
=c_{1}/\sqrt{1+y^{2}}$ and $\overline{B}\left( y\right) =c_{2}$ so that the
reduced section is%
\[
\overline{X}\left( y\right) =\left( y,\frac{c_{1}}{\sqrt{1+y^{2}}},c_{2},\frac{%
c_{1}y}{\sqrt{1+y^{2}}}\right). 
\]
Let us now construct a vector field $X$ according to the one-to-one correspondence in (\ref{reconstrucion23}). An easy computation shows that the dynamics generated by $X$ is 
\begin{equation} \label{ex-1-HJ}
\dot{x}=c_{1}/\sqrt{1+y^{2}},\qquad \dot{y}=c_{2},\qquad \dot{%
z}=c_{1}y/\sqrt{1+y^{2}}. 
\end{equation}
See that $X$ solves the Hamilton-Jacobi problem in the unreduced picture (\ref{HJ-ex-1-unreduced}). Compare the system in (\ref{ex-1-HJ}) with the one in \cite{de2014hamilton}.

\subsection{The two-wheeled carriage}

In this subsection, we shall write the Hamilton-Jacobi equation for both of reduced and unreduced pictures of a two-wheeled carriage \cite{Koiller,Neimark}. The nonholonomic character of this example is investigated in \cite{LeonDMDD}. Later, it is discussed in the framework of geometrization of generalized
Chaplygin systems in \cite{Cap1}. In a more recent work \cite{deLeMadeDi10}, the Hamilton-Jacobi formalism of this example is studied with linear almost Poisson structures. In the present work, in addition to previous studies, we shall investigate the Hamilton-Jacobi formalism of the reduced nonholonomic system. To this end, we start by introducing the basic ingredients that describe the system.

\begin{figure}[h]
  \centering
   \includegraphics[width=0.5\textwidth]{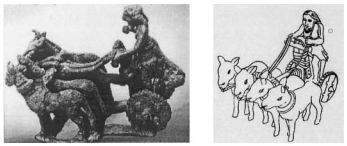} 
   \caption{The copper quadriga from Tell Agrab. A two-wheeled carriage drawn by four oxen (or onagers). Height: 7cm. Iraqi Museum, Baghdad}
\end{figure}

\newpage

The configuration space $Q$ of the two-wheeled carriage is the product manifold $SE(2)\times \mathbb{T}^{2}$ where $SE(2)$ is the special Euclidean group of $\mathbb{R}^{2}$, physically representing the rigid body motion in the $xy-$plane, with local coordinates $\left( x,y,\varphi \right) $, whereas $\mathbb{T}^{2}$ is a two torus, realizing the left and right wheel of the carriage, with local coordinates $( \phi _{1},\phi _{2}) $.\\

\begin{figure}[h!]
  \centering
   \includegraphics[width=0.4\textwidth]{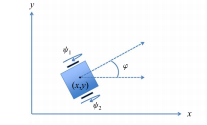}
   \caption{Two-wheeled carriage}
\end{figure}

For the fiber coordinates of the tangent bundle $TQ$, we consider the induced coordinates $( \dot{x},\dot{y},\dot{\varphi},\dot{\phi}_{1},\dot{\phi}_{2})$ and the Lagrangian function generating the equation of motion is 
\[
L=\frac{1}{2}m\left( \dot{x}^{2}+\dot{y}^{2}\right) +m_{0}l\dot{\varphi}%
\left( \dot{y}\cos \varphi -\dot{x}\sin \varphi \right) +\frac{1}{2}J\dot{%
\varphi}^{2}+\frac{1}{2}C\left( \dot{\phi}_{1}^{2}+\dot{\phi}_{2}^{2}\right) 
\]
where $m_{0}$ is the mass of the body without wheels, $m$ is the mass of the
total system, $J$ is the moment of inertia when it rotates as a whole around
the vertical axis passing through the point $(x,y)$, and $C$ the axial moment
of inertia, $l$ is the distance of the point $\left( x,y\right) $ to the
center of mass. In order to impose no lateral sliding and no sliding of the wheels, we introduce the following set of linear nonholonomic constraints%
\begin{eqnarray*}
\psi ^{1} &=&\dot{x}+\frac{a\cos \varphi }{2}\dot{\phi}_{1}+\frac{a\cos
\varphi }{2}\dot{\phi}_{2}, \\
\psi ^{2} &=&\dot{y}+\frac{a\sin \varphi }{2}\dot{\phi}_{1}+\frac{a\sin
\varphi }{2}\dot{\phi}_{2}, \\
\psi ^{3} &=&\dot{\varphi}+\frac{a}{2r}\dot{\phi}_{1}-\frac{a}{2r}\dot{\phi}%
_{2},
\end{eqnarray*}%
where $a$ is the radius of the wheels, and $r$ is a half of the length of the
axle. These constraint functions determine a $7$-dimensional constraint submanifold  
\[
N_{\ell }=\left\{ v\in TQ:\psi^{a}\left( v\right) =0,\  a=1,2,3\right\}
\]
equipped with a local coordinate chart $( x,y,\varphi ,\phi _{1},\phi _{2},\dot{\phi}_{1},\dot{\phi}_{2}) $.

Let us analyze the nonholonomic character of the submanifold $N_{\ell}$. It is immediate now to define a set of constraint one-forms
\begin{eqnarray}
\bar{\psi}^{1} &=&dx+\frac{a\cos \varphi }{2}d\phi _{1}+\frac{a\cos \varphi }{2%
}d\phi _{2},  \nonumber \\
\bar{\psi}^{2} &=&dy+\frac{a\sin \varphi }{2}d\phi _{1}+\frac{a\sin \varphi }{2%
}d\phi _{2},  \nonumber \\
\bar{\psi}^{3} &=&d\varphi +\frac{a}{2r}d\phi _{1}-\frac{a}{2r}d\phi _{2}
\end{eqnarray}%
spanning the codistribution $N_{\ell }^{o}$. From this, we can compute the constraint submanifold $%
N_{\ell }$ as the space generated by the two vector fields $\xi_1 $and $\xi_2$,
$$N_{\ell }=\left \langle \{\xi _{1},\xi _{2}\} \right \rangle $$
where  
\begin{eqnarray*}
\xi _{1} &=&-\frac{a\cos \varphi }{2}\dfrac{\partial}{\partial x}-\frac{a\sin \varphi }{2}%
\dfrac{\partial}{\partial y}-\frac{a}{2r}\dfrac{\partial}{\partial \varphi}+\dfrac{\partial}{\partial \phi_{1}}, \\
\xi _{2} &=&-\frac{a\cos \varphi }{2}\dfrac{\partial}{\partial x}-\frac{a\sin \varphi }{2}%
\dfrac{\partial}{\partial y}+\frac{a}{2r}\dfrac{\partial}{\partial \varphi}+\dfrac{\partial}{\partial \phi_{2}}.
\end{eqnarray*}%
The Jacobi-Lie bracket of these two vector fields is computed and labelled as 
$\xi _{3}$,
\[
\xi _{3}=\left[ \xi _{1},\xi _{2}\right] =-\frac{a^{2}\sin \varphi }{2r^{2}}%
\dfrac{\partial}{\partial x}+\frac{a^{2}\cos \varphi }{2r^{2}}\dfrac{\partial}{\partial y}.
\]
We compute the Jacobi-Lie bracket of $\xi _{1}$ and $\xi _{3}$, and thus arrive at 
\[
\xi _{4}=\left[ \xi _{1},\xi _{3}\right] =\frac{a^{3}\cos \varphi }{4r^{2}}%
\dfrac{\partial}{\partial x}+\frac{a^{3}\sin \varphi }{4r^{2}}\dfrac{\partial}{\partial y}.
\]%
These four vector fields $\xi_1$, $\xi_2$, $\xi_3$ and $\xi_4$ are linearly independent. The dimension of the tangent space $T_qQ$ at each $q$ in $Q$ is five. In order to find a fifth
linear independent vector field, we compute the following Jacobi-Lie brackets 
\[
\left[ \xi _{1},\xi _{4}\right] =-\frac{a^{2}}{4r}\xi _{3},\text{ \ \ }\left[
\xi _{2},\xi _{3}\right] =-\xi _{4},\text{ \ \ and \ \ }\left[ \xi _{3},\xi
_{4}\right] =0.
\]%
It is easy now to observe that the iteration of brackets will not give rise to a fifth linearly independent vector field. Notice that
\begin{eqnarray*}
\left[ \xi_{2},\xi_{4} \right] &=& \ \  \ \left[ \xi_2 ,\left[ \xi _{1},\xi _{3}\right] \right] \\
&=&-\left[ \xi_3 ,\left[ \xi _{2},\xi _{1}\right] \right]  - \left[ \xi_1 ,\left[ \xi _{3},\xi _{2}\right] \right] \\
&= &-\left[ \xi_3, - \xi_3 \right] - \left[  \xi_{1} , \xi_{4} \right]\\
&=& \ \ \ \frac{a^{2}}{4r}\xi _{3}.
\end{eqnarray*}

 This implies that the constraint submanifold $%
N_{\ell }$ is not completely nonholonomic. In this case, the Hamilton-Jacobi equation of the system is obtained by means of Theorem \ref{nhhj21}. It reads:
\begin{equation}
d\left( E_{L}\circ X\right) \in N_{\ell }^{o}  \label{unRedHJ-ex}
\end{equation}
with the requirement that $X(Q) \subset N_{\ell}$ and $d(\mathbb{F}L \circ X)$ be in ${\mathcal I} (N_{\ell}^0)$.\\
Instead of finding a solution $X$ to this problem, we now employ the reduction procedure tofind a solution and the lift solutions of the Hamilton-Jacobi problem of the reduced system. 

Consider a Lie group action $G=SE\left( 2\right) $ on $Q=SE\left(
2\right) \times \mathbb{T}^{2}$ is defined by group multiplication of $G$ to the first factor $SE\left(
2\right)$ in $Q$ while acting trivially on the second factor $\mathbb{T}^{2}$.
Notice that the two-wheeled carriage is a \u{C}aplygin system. In this case, the reduced dynamics can be written as \cite{cantrijn1999reduction}  (see Eq. (\ref{128}))
\[
\iota_{\overline{\Gamma}_{L,N_{\ell }}}\omega_{L^{*}} = d E_{L^{*}} - \overline{\alpha},
\]%
where the reduced Lagrangian is computed to be
\[
L^{*}\left( \phi _{1},\phi _{2},\dot{\phi}_{1},\dot{\phi}_{2}\right) =%
\frac{1}{8}ma^{2}\left( \dot{\phi}_{1}+\dot{\phi}_{2}\right) ^{2}+\frac{%
Ja^{2}}{8r^{2}}\left( \dot{\phi}_{2}-\dot{\phi}_{1}\right) ^{2}+\frac{1}{2}%
C\left( \dot{\phi}_{1}^{2}+\dot{\phi}_{2}^{2}\right) .
\]%
whereas the one-form $\overline{\alpha}$ is 
\[
\overline{\alpha}=\frac{a^{3}m_{o}l}{4r^{2}}\left( \dot{\phi}%
_{2}-\dot{\phi}_{1}\right) \dot{\phi}_{2}d\phi _{1}-\frac{a^{3}m_{o}l}{4r^{2}%
}\left( \dot{\phi}_{1}-\dot{\phi}_{2}\right) \dot{\phi}_{1}d\phi _{2}.
\]%
If we apply the Hamilton-Jacobi theorem \ref{HJ-red-Cap} to this case, we arrive at the following Hamilton-Jacobi equation
\[
d\left( E_{L^{*}}\circ \overline{X}\right) =-\overline{X}^{\ast }%
\overline{\alpha}.
\]
The energy function $E_{L^{*}}$ corresponding to the reduced
Lagrangian $L^{*}$ is equal to $L^{*}$. A straightforward calculation proves that 
two solutions of this equation are 
\[
\overline{X}_{1}=e^{\frac{K}{R}\phi _{2}}\frac{\partial }{\partial \phi _{1}}%
\text{, \ \ }\overline{X}_{2}=e^{-\frac{K}{R}\phi _{1}}\frac{\partial }{%
\partial \phi _{2}}
\]%
where $K=m_{o}la^{3}/4r^{2}$ and $R=\frac{1}{4}ma^{2}+Ja^{2}/4r^{2}+C$.\\
We compute the horizontal lifts to the manifold $Q$ of the solutions of the reduced equation, by means of a connection \cite{LeonDMDD} 
\begin{eqnarray*}
\left( \dfrac{\partial}{\partial \phi_{1}}\right) ^{h} &=&\dfrac{\partial}{\partial \phi_{1}}-a\cos
\varphi \dfrac{\partial}{\partial x}-a\sin \varphi \dfrac{\partial}{\partial y}-\frac{a}{r}\dfrac{\partial}{\partial \varphi}, \\
\left( \dfrac{\partial}{\partial \phi_{2}}\right) ^{h} &=&\dfrac{\partial}{\partial \phi_{2}}-a\cos
\varphi \dfrac{\partial}{\partial x}-a\sin \varphi \dfrac{\partial}{\partial y}+\frac{a}{r}\dfrac{\partial}{\partial \varphi}.
\end{eqnarray*}%
We retrieve two solutions of the unreduced Hamilton-Jacobi problem (\ref%
{unRedHJ-ex}) that read
\begin{eqnarray*}
X_{1} &=&e^{\frac{K}{R}\phi _{2}}\left( \dfrac{\partial}{\partial \phi_{1}}-a\cos \varphi
\dfrac{\partial}{\partial x}-a\sin \varphi \dfrac{\partial}{\partial y}-\frac{a}{r}\dfrac{\partial}{\partial \varphi}\right),  \\
X_{2} &=&e^{-\frac{K}{R}\phi _{1}}\left( \dfrac{\partial}{\partial \phi_{1}}-a\cos \varphi
\dfrac{\partial}{\partial x}-a\sin \varphi \dfrac{\partial}{\partial y}-\frac{a}{r}\dfrac{\partial}{\partial \varphi}\right). 
\end{eqnarray*}

\section{Conclusions}
In this extensive paper we present a survey article that also includes new results: some are a generalization of the previous results collected in the survey, and others are not only generalizations, but brand new results.
The survey consists of a review of the most relevant aspects in the geeometric description of nonholonomic mechanical systems with symmetries. In particular, we recall the available Hamilton--Jacobi theories for nonholonomic systems with linear constraints (still without considering symmetries) and we here extend it to the case including nonlinear constraints, this is mainly included in Theorem 10. From there, we obtain two consequences: one occurs when the system does not produce forces, this is a curious fact that is commented in Collorary 1, although its usefulness is unexplored.
The second consequence or Collorary 2 is that we can make use of Theorem 6, the so called Chow-Rashevskii theorem to achieve an equation instead of an inequation. This is a generalization of the results in the linear case by T. Ohsawa and A. Bloch in \cite{ToOsBlo} to the nonlinear case.

 This introductory review together with our generalization serves us to compile a reduction of a Hamilton--Jacobi theory for nonholonomic systems with symmetries, very generally, also including the nonlinear case. The most general case is summarized in Theorem 18, and we stress the properties of two particular cases: the kinematic case in Theorem 22, and the horizontal case in Theorem 30, that are inspected in depth. Reduced versions for Colloraries 1 and 2 are presented.
 Subsequently, an specific method for reconstruction of solutions starting from the reduced solution is provided in equation (94).
 It is important to notice that these theorems have been both redacted in the Lagrangian and Hamiltonian counterparts.
 
 These are all novel results in our research, and to illustrate these newfangled development, we propose two examples. 

\section*{Acknowledgements}
This work has been partially supported by MINECO Grants MTM2016-76-072-P and the ICMAT Severo Ochoa projects SEV-2011-0087 and SEV-2015-0554. V.M.~Jim{\'e}nez wishes to thank MINECO for a FPI-PhD Position and the referee for the suggestions. We would like to thank the hospitality of J. {\'S}niaticky and the University of Victoria in Canada for hosting us to carry out the wisely discussions that turned out to complete this paper.

\bibliographystyle{plain}
\bibliography{Library}
\end{document}